\documentclass[preprint,12pt]{elsarticle}

\if()
\usepackage[paperwidth=9.2in,paperheight=11.7in]{geometry}
\usepackage{authblk}
\usepackage{algorithm}
\usepackage{algorithmic}
\usepackage{subfigure}
\usepackage{type1cm}
\usepackage[dvipdfmx]{graphicx}
\usepackage{wrapfig}
\usepackage{latexsym}
\usepackage{amsmath}
\usepackage{algorithm}
\usepackage{algorithmic}
\usepackage{multicol}
\fi
 
\usepackage{multicol}
\usepackage{multirow}

\usepackage{amssymb}
\usepackage{graphicx}
\usepackage{wrapfig}
\usepackage{latexsym}
\usepackage{amsmath}
\usepackage[linesnumbered,ruled,vlined]{algorithm2e}
\DontPrintSemicolon
\SetKwInput{Main}{Main Routine of Agent $a_i$}
\SetKwInput{MainMinor}{Main Routine of Agent $a_{i'}$}
\SetKwInput{ProcedureCheck}{{\bf Procedure} \textit{Check}()}
		
\usepackage{multicol}
\usepackage{amsthm}
\usepackage{float}
\usepackage{type1cm} 
\usepackage{subfigure}

\newtheorem{theorem}{Theorem}
\newtheorem{lemma}{Lemma}
\newtheorem{definition}{Definition}

\journal{Theoretical Computer Science}

\begin{document}
	
\begin{frontmatter}
		
\title{Partial gathering of mobile agents in dynamic rings\tnoteref{t1,t2}}
\tnotetext[t1]{The conference version of this paper is published in the proceedings of 
22nd International Symposium on Stabilization, Safety, and Security of Distributed Systems
(SSS 2021).}
\tnotetext[t2]{This work was partially supported by JSPS KAKENHI Grant Number 18K18029, 18K18031, 20H04140, 20KK0232, 21K17706, and 22K11971;
	the Hibi Science Foundation; and Foundation of Public Interest of Tatematsu.}

\author[a]{Masahiro~Shibata\corref{cor1}}
\ead{shibata@csn.kyutech.ac.jp}
\author[b]{Yuichi~Sudo}
\ead{sudo@hosei.ac.jp}
\author[c]{Junya~Nakamura}
\ead{junya@imc.tut.ac.jp}
\author[d]{Yonghwan~Kim}
\ead{kim@nitech.ac.jp}

\cortext[cor1]{Corresponding author. Tel.:+81-9-4829-7656.}

\address[a]{Graduate school of Computer Science and Systems Engineering, 
	Kyushu Institute of Technology\\
	680-4 Kawazu, Iizuka, Fukuoka 820-8502, Japan}
\address[b]{Graduate School of Computer and Information Sciences, Hosei University\\ 
	3-7-2 Kajino-cho, Koganei-shi, Tokyo, 184-8584, Japan}
\address[c]{Information and Media Center, Toyohashi University of Technology\\ 
		1-1 Hibarigaoka, Tempaku-cho, Toyohashi, Aichi, 441-8580, Japan}
\address[d]{Graduate School of Computer Science and Engineering, Nagoya Institute of Technology\\ 
Gokiso-cho, Showa-ku, Nagoya, Aichi, 466-8555 Japan}

\begin{abstract}
	
In this paper, 
we consider the partial gathering problem of mobile agents in synchronous dynamic bidirectional ring networks.
The partial gathering problem is a  generalization
of the (well-investigated) total gathering problem, 
which requires that all $k$ agents distributed in the network 
terminate at a non-predetermined single node.
The partial gathering problem  requires, 
for a given positive integer $g\,(< k)$,  that 
agents terminate in a configuration such that 
either at least $g$ agents or no agent exists at each node.
When $k\ge 2g$,
the requirement for the partial gathering problem is strictly weaker   
than that for the  total gathering problem, and thus
it is interesting to clarify the difference in the move complexity between them.
So far, the partial gathering problem 
has been considered in static graphs.
In this paper, we start considering partial gathering in dynamic graphs. 
As a first step, we consider this problem in 1-interval connected rings, that is, 
one of the links in a ring may be missing at each time step.  
In such networks, 
focusing on the relationship between the values of $k$ and $g$,
we fully characterize the solvability of the partial gathering problem
and analyze the move complexity of the proposed algorithms when the problem can be solved. 	
First, we show that the $g$-partial gathering problem is unsolvable when $k \le 2g$.
Second, we show that
the problem can be solved with
$O(n\log g)$ time and the total number of $O(gn\log g)$ moves when $2g+1\le k \le 3g-2$.
Third, we show that
the problem can be solved with
$O(n)$ time and the total number of $O(kn)$ moves  when $3g-1\le k \le 8g-4$.
Notice that since $k = O(g)$ holds when $3g-1 \le k \le 8g-4$, 
the move complexity $O(kn)$ in this case can be represented also as $O(gn)$.
Finally, we show that
the problem can be solved with
$O(n)$ time and the total number of $O(gn)$ moves when $k\ge 8g-3$.
These results mean that the partial gathering problem can be solved also in dynamic rings
when $k\ge 2g+1$.
In addition, 
agents require a total number of $\Omega(gn)$ (resp., $\Omega(kn)$) moves 
to solve the partial (resp., total) gathering problem.
Thus, 
when $k\ge 3g-1$, agents can solve the partial gathering problem with the asymptotically optimal total number of 
$O(gn)$ moves, which is strictly smaller than that for the total gathering problem.	\\

\end{abstract}

\end{frontmatter}  

\section{Introduction}

\subsection{Background and Related Work}
A distributed system
comprises a set of computing entities ({\em nodes}) connected by communication links.
As a promising design paradigm of distributed systems, 
(mobile) agents have attracted much
attention 
\cite{system,motivation1}.
The agents can 
traverse the system, carrying information collected at visited nodes,
and execute an action at each node using the information
to achieve a task.
In other words, agents can
encapsulate the process code and data, which simplifies the design of distributed systems 
\cite{good,motivation2}.

The {\em total gathering problem} (or the rendezvous problem)
is a fundamental problem for agents' coordination.
When a set of $k$ agents are  arbitrarily placed at nodes,
this  problem requires that all the $k$ agents 
terminate  at a non-predetermined single node. 
By meeting at a single node, 
all agents can share information or synchronize their behaviors.
The total gathering problem has been considered in various kinds of networks
such as rings \cite{gatheringBook,token1,token2,Mr.Kawai},
trees \cite{littleMemory,Tree3}, tori \cite{token3},  
and arbitrary networks \cite{rvArbitrary,rvPolynominal}.

Recently, a variant of the total gathering problem, 
called the {\em $g$-partial gathering problem} \cite{PartialRing}, has been considered.
This problem does not require all agents to meet at a single node, 
but allows agents to meet at several nodes separately.  
Concretely, for a given positive integer  $g\,(< k)$, 
this problem requires that agents terminate in a configuration such that 
either at least $g$ agents or no agent exists at each node.
From a practical point of view, 
the $g$-partial gathering problem is still useful especially in large-scale networks.
That is, 
when $g$-partial gathering is achieved, 
agents are partitioned into groups each of which has at least $g$ agents,
each agent can share information and tasks with agents in the same group,
and each group can partition the network and 
then patrol its area that it should monitor efficiently. 
\if()
each group with at least $g$ agents can partition the network and 
then patrol its area that they should monitor efficient
each agent can share information and tasks with at least $g-1$ agents staying at the same node.
In addition, while in total  gathering all the agents meet at a single node, 
in $g$-partial gathering agents meet at multiple nodes separately. 
This means that, after achieving $g$-partial gathering,
each group with at least $g$ agents can partition the network and 
then patrol its area that they should monitor efficiently. 
\fi
The $g$-partial gathering problem is interesting 
also from a theoretical point of view.
Clearly, if $k <2g$ holds, 
the $g$-partial gathering problem is equivalent to the total  gathering problem. 
On the other hand, 
if $k\ge 2g$ holds, 
the requirement for the $g$-partial gathering problem is strictly weaker than 
that for the total gathering problem. 
Thus,  there exists possibility that 
the $g$-partial gathering problem can be solved with strictly smaller total number of moves 
(i.e., lower costs)
compared to the total gathering problem.

As related work, 
Shibata et al. considered the $g$-partial gathering problem in rings \cite{PartialRing,Kawata,partialSebastien},
trees \cite{PartialTree}, and arbitrary networks \cite{partialArbitrary}.
In \cite{PartialRing,Kawata}, they considered it 
in unidirectional ring networks with whiteboards (or memory spaces
that agents can read and write) at nodes.
They mainly showed that, if agents have distinct IDs and the algorithm is deterministic,  
or if agents do not have distinct IDs
and the algorithm is randomized, 
agents can achieve $g$-partial gathering with the total number of $O(gn)$ moves (in expectation),
where $n$ is the number of nodes. 
Notice that in the above results 
agents do not have any global knowledge such as $n$ or $k$.
In \cite{partialSebastien}, they considered $g$-partial gathering for another mobile entity called \textit{mobile robots}
that have no memory but can observe all nodes and robots in the network. 
In the case of using mobile robots, 
they also showed that $g$-partial gathering can be achieved with the total number of $O(gn)$ moves. 
\if()
considered two problem settings about agents:
distinct agents (i.e., agents with distinct IDs) and 
anonymous agents (i.e., agents without IDs) with knowledge of $k$. 
For distinct agents, 
they gave a deterministic algorithm to solve
the $g$-partial gathering problem in $O(gn)$ total number of moves,
where $n$ is the number of nodes.
For anonymous agents with knowledge of $k$,
they considered deterministic and randomized cases.
In the deterministic case, they showed that 
there exist unsolvable initial configurations.
In addition, they gave an algorithm to solve the problem 
from any  solvable initial configuration 
in $O(kn)$ total number of moves.
In the randomized case, 
they gave an algorithm to solve the problem  
in $O(gn)$ expected total number of moves.
\fi
In addition, the $g$-partial (resp., the total) gathering problem in ring networks 
requires a total number of $\Omega (gn)$ (resp., $\Omega (kn)$) moves.
Thus, the above results are asymptotically optimal in terms of the total number of moves, 
and the 
number $O(gn)$ 
is strictly smaller than that for the total gathering problem when $g=o(k)$. 
In tree and arbitrary networks, 
they also proposed algorithms to solve the $g$-partial gathering problem 
with strictly smaller total number of moves compared to 
the total gathering problem for some settings.

While all the above work on the total gathering problem and the $g$-partial gathering problem 
are considered in \textit{static graphs} where a network topology does not  change during an execution,
recently many problems involving agents have been studied 
in \textit{dynamic graphs}, where  a topology changes during an execution.
For example, 
the total gathering problem \cite{gatheringDynamicRing},
the exploration problem \cite{explorationDynamicRing,explorationDynamicTori}, 
the compact configuration problem \cite{compacting},
the patrolling problem \cite{patrollingDynamicRing}, 
and the uniform deployment problem \cite{UniformDynamicRing}
are considered in dynamic graphs.
In \cite{gatheringDynamicRing}, Luna et al. considered the total  gathering problem in 
\textit{1-interval connected rings}, that is, one of the links in a ring may be missing at each time step.
In such networks, they considered the trade-off between the existence of agent ability (chirality and cross detection) and the algorithm performances. 
However, to the best of our knowledge,  
there is no work 
for  $g$-partial gathering in dynamic graphs.

\begin{figure}[!t] 
	\centering 
	\includegraphics[scale=0.5]{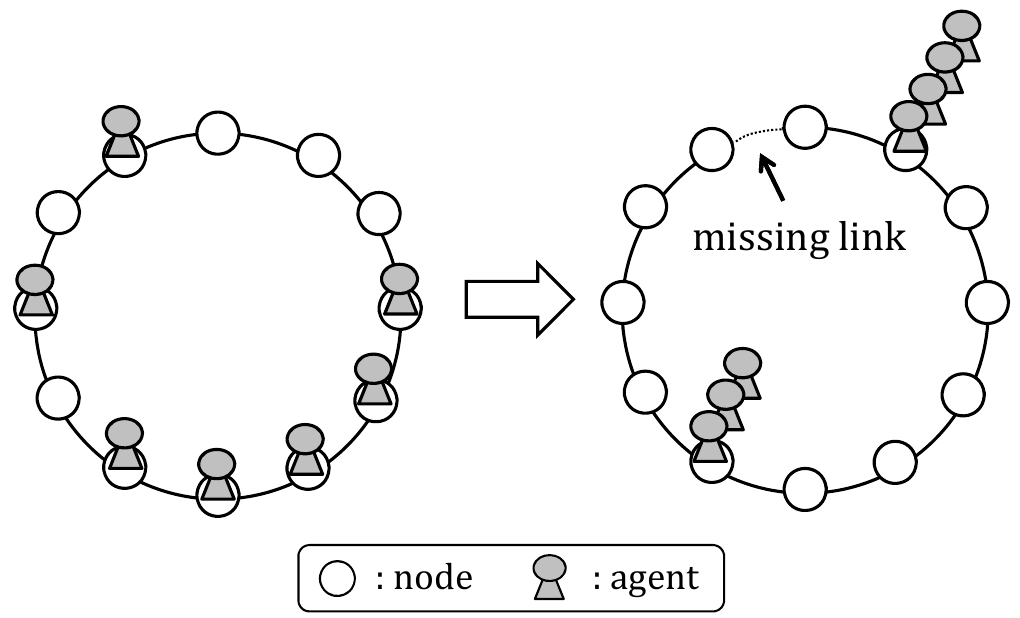} 
	\caption{An example of the $g$-partial gathering problem in a dynamic ring ($g=3$).}
	\label{fig:example}
\end{figure} 
 
\subsection{Our Contribution}
\label{sec:contribution} 
In this paper,
we consider the $g$-partial gathering problem of mobile agents in 1-interval connected  rings \cite{gatheringDynamicRing,explorationDynamicRing,compacting,UniformDynamicRing}.
An example is given in Fig.\,\ref{fig:example}.
We assume that the ring is bidirectional and each node has a whiteboard. 
In addition, we assume that agents have distinct IDs, chirality,  knowledge of $n$ and $k$, and behave fully synchronously.
In such settings, 
focusing on the relationship between the values of $k$ and $g$,
we fully characterize the solvability of the $g$-partial gathering problem
and analyze the time and move complexities of the proposed algorithms when the problem can be solved.

We summarize our results and the existing result \cite{PartialRing} for agents with distinct IDs in static rings in Table \ref{table:result}.
First, we show that the $g$-partial gathering problem is unsolvable when $k \le 2g$.
Second, we show that
the problem can be solved with
$O(n\log g)$ time and the total number of $O(gn\log g)$ moves when $2g+1\le k \le 3g-2$.
Third, we show that
the problem can be solved with
$O(n)$ time and the total number of $O(kn)$ moves  when $3g-1\le k \le 8g-4$.
Notice that since $k = O(g)$ holds when $3g-1 \le k \le 8g-4$, 
the move complexity $O(kn)$ in this case can be represented also as $O(gn)$.
Finally, we show that
the problem can be solved with
$O(n)$ time and the total number of $O(gn)$ moves when $k\ge 8g-3$.
These results mean that the $g$-partial gathering problem can be solved also in dynamic rings  when $k\ge 2g+1$.
In particular, when $k\ge 3g-1$, the time complexity $O(n)$ of the proposed algorithms is asymptotically optimal and 
the total number of $O(gn)$ moves is also asymptotically optimal, 
which is strictly smaller than the move complexity for the total gathering problem.
Moreover, it is worthwhile to mention that,
total gathering, which is equivalent to $g$-partial gathering with $k<2g$, 
cannot be solved in 1-interval connected rings~\cite{gatheringDynamicRing}.
Thus, Luna et al.~\cite{gatheringDynamicRing} solve the total gathering problem
by relaxing the problem specification:
they regard that gathering is achieved when every agent stays at one of the two neighboring nodes.
On the other hand, $g$-partial gathering with $k>2g$, 
which we address in this paper, does not require such relaxation.

\begin{table*}[t!]
	\centering 
	\newlength{\myheight}
	\setlength{\myheight}{0.8mm}
	\caption{Results of $g$-partial gathering for agents with distinct IDs in ring networks for the case of $k\ge 2g$ ($n$: \#nodes, $k$: \#agents).}
	\label{table:result}
	\scriptsize

	\begin{tabular}{|c||c|c|c|c|c|}
		\hline 
		
		\rule{0cm}{\myheight} &\multirow{2}{*}{Result in \cite{PartialRing}} & 
		\multicolumn{4}{|c|}{Results of this paper}\\  \cline{3-6} 
		\rule{0cm}{\myheight} & &Result 1 (Sec. \ref{sec:2g}) & Result 2 (Sec. \ref{sec:more2g})& Result 3 (Sec. \ref{sec:more3g}) & Result 4 (Sec. \ref{sec:more8g})\\		
		\hline
		\hline
		\rule{0cm}{\myheight} Static/Dynamic ring &  Static & Dynamic & Dynamic & Dynamic & Dynamic\\ 
		\hline
		\rule{0cm}{\myheight} Knowledge of $n$ and $k$ & No &Available& Available & Available & Available   \\ 
		\hline	
		\rule{0cm}{\myheight} Relation between $k$ and $g$ & No condition   &$k \le 2g$& $2g+1\le k \le 3g-2$   & $3g-1\le k\le 8g-4$ & $k\ge 8g-3$ \\
		\hline		
		\rule{0cm}{\myheight} Solvable/Unsolvable
		& Solvable & Unsolvable & Solvable & Solvable & Solvable\\
		
		\hline		
		\rule{0cm}{\myheight} Time complexity 
		&$\Theta (n)$& - & $O(n\log g)$ & $\Theta(n)$ & $\Theta (n)$ \\
		\hline
		\rule{0cm}{\myheight} Agent moves  
		&$\Theta (gn)$ & - & $O(gn\log g)$ &$O(kn) (= O(gn))$ & $\Theta (gn)$ \\
		\hline
	\end{tabular}
	\footnotesize
\end{table*}

\section{Preliminaries}
\label{model}

\subsection{System Model}
\label{sec:model}
We basically follow the model 
defined in 
\cite{gatheringDynamicRing}.
A {\em dynamic bidirectional ring} $R$ is defined as 2-tuple $R=(V,E)$, 
where $V=\{v_{0},v_{1}, \ldots ,v_{n-1}\}$ is a set of $n$ nodes 
and $E=\{e_{0},e_{1},\ldots ,e_{n-1}\}$ $(e_{i}$=$\{v_{i}$, $v_{(i+1) \bmod n}\})$ is a set of links. 
For simplicity, 
we denote $v_{(i+j)\mod n}$ (resp., $e_{(i+j)\mod n}$) by 
$v_{i+j}$ (resp., $e_{i+j}$). 
We define the direction from $v_i$ to $v_{i+1}$ (resp., $v_i$ to $v_{i-1}$)
as the {\em forward} or \textit{clockwise}
(resp., \textit{backward} or \textit{counterclockwise}) direction.
In addition,
one of  the links in a  ring may be missing at each time step, and
the missing link is chosen by an \textit{adversarial scheduler}.
Such a dynamic ring is known as a \textit{1-interval connected ring}. 
The {\em  distance} from node $v_{i}$ to $v_{j}$ is defined to be 
$\min\{(j-i)\mod n, (i-j)\mod n\}$. 
Notice that this definition of the distance is correct when 
no link from $v_i$ to $v_j$ (resp., $v_j$ to $v_i$) is missing
when the distance is $(j-i) \mod n$ (resp., $(i-j) \mod n$).
Moreover, we assume that nodes are anonymous, i.e., they do not have IDs. 
Every node $v_i\in V$ has a whiteboard that  
agents at node $v_i$ can read from and write on. 

Let $A=\{a_{0},a_{1},\ldots$ $,a_{k-1}\}$ be a set of $k\,(\leq n)$ agents.
Agents can move through directed links, that is, 
they can move from $v_i$ to $v_{i+1}$ (i.e., move forward) 
or from $v_i$ to $v_{i-1}$ (i.e., move backward) for any $i$.
Agents have distinct IDs and knowledge of $n$ and $k$.
Agents have \textit{chirality}, that is, they agree on the orientation of clockwise and counterclockwise
direction in the ring.
In addition, agents cannot detect
whether other agents exist at the current node or not.
An agent $a_i$ is defined as a deterministic finite automaton
$(S$, $W$, $\delta$, $s_{\textit{initial}}$, $s_{\textit{final}}$, 
$w_\textit{initial}$, $w'_\textit{initial})$. 
The first element $S$ is the set of all states of an agent, including two special states,
initial state $s_{\textit{initial}}$ and final state $s_{\textit{final}}$.
The second element $W$ is the set of all states (contents) of a whiteboard, 
including two special initial states $w_\textit{initial}$ and $w'_\textit{initial}$.
We explain $w_\textit{initial}$ and $w'_\textit{initial}$ in the next paragraph.
The third element $\delta: S\times W\mapsto S\times W\times M$ is the state transition function
that decides, from the current state of $a_i$ and the current node's whiteboard,
the next states of $a_i$ and the whiteboard,
and whether $a_i$ moves to its neighboring node or not.
The last element  $M=\{-1,0,1\}$ in $\delta$ 
represents whether $a_i$ makes a movement or not. 
The value $1$  (resp., $-1$) means moving forward (resp., backward) and 0 means staying at the current node. 
We assume that $\delta \,(s_\textit{final}, w_j)=(s_\textit{final},w_j,0)$ holds
for any state $w_j\in W$, which means that 
$a_i$ never changes its state, updates the contents of a whiteboard, or 
leaves the current node once it reaches state $s_\textit{final}$.
We say that an agent \textit{terminates} when its state changes to $s_\textit{final}$.
Notice that $S, \delta, s_\textit{initial}$, and $s_\textit{final}$ can be dependent on the agent's ID.


In an agent system, (global) {\em configuration} $c$ is defined as a product of
the states of all  agents, the states (whiteboards' contents) of all  nodes, and 
the locations (i.e., the current nodes) of all  agents.
\if()
In an agent system, a (global) {\em configuration} is defined as the Cartesian product 
$\cal{S} \times \cal{W} \times  \cal {L}$, where $\cal{S}\in$ $S^k$ represents the states of all agents,
$\cal{W}\in$ $W^n$ represents the states (whiteboard' contents) of all nodes, and
$\cal{L}\in$ $\{0,1,\ldots n-1\}^k$ represents the current locations of agents. 
The locations of agents $\cal{L}$ =$(l_0, l_1, \ldots l_{k-1})$ implies that each
agent $a_i$ is located at node $v_{l_i}$.
\fi
We define $C$ as a set of all configurations. 
In an initial configuration $c_0\in C$, we assume that 
agents are deployed arbitrarily at mutually distinct nodes (or no two agents start at the same node), and
the state of each whiteboard is $w_\textit{initial}$ or $w'_\textit{initial}$
depending on the existence of an agent.
That is, when an agent exists at node $v$ in the initial configuration,
the initial state of $v'$s whiteboard is $w_\textit{initial}$.
Otherwise, the state is $w'_\textit{initial}$.

During an execution of the algorithm, 
we assume that agents move instantaneously, that is, 
agents always exist at nodes (do not exist on links). 
Each agent  executes the following four operations in an \textit{atomic action}:
1) reads the contents of  its current  node's whiteboard, 
2) executes local computation (or changes its state), 
3) updates the contents of  the current node's whiteboard, and 
4) moves to its neighboring node or stays at the current node.
If several agents exist at the same node,
they take atomic actions interleavingly in an arbitrary order.
In addition, when an agent tries to move to its neighboring node (e.g., from node $v_j$ to $v_{j+1}$) 
but the corresponding link (e.g., link $e_j$) is missing, we say that 
the agent is \textit{blocked},
and it still exists at $v_j$ at the beginning of the next atomic action.

In this paper, we consider a \textit{synchronous execution}, that is, 
in each time step called \textit{round},
all agents perform 
atomic actions.
Then, an \textit{execution} 
starting 
from $c_0$ is defined as $E = c_0,c_1, \ldots$
where each $c_i\ (i \ge 1)$ is the configuration reached from $c_{i-1}$
by atomic actions of all agents.
An execution is infinite, or ends in 
a \textit{ final configuration}  where
the state of every agent is $s_\textit{final}$.

\subsection{The Partial Gathering Problem}

The requirement for the partial gathering problem is that, 
for a given integer $g$, 
agents terminate in a configuration such that 
either at least $g$ agents or no agent exists at each node.
Formally, we define the $g$-partial gathering problem as follows.

\begin{definition}\label{teigi:partial}
	An algorithm solves the $g$-partial gathering problem in dynamic rings when the following conditions hold:
	\begin{itemize}
		\item Execution $E$ is finite (i.e., all agents terminate in state $s_\textit{final}$).
		\item In the final configuration,
		at least $g$ agents exist at any node where an agent exists.
	\end{itemize}
\end{definition}

In this paper, we evaluate the proposed algorithms by  the time complexity 
(the number of rounds required for agents to solve the problem)
and the total number of agent moves.
In \cite{PartialRing}, the lower bound on the total number of agent moves for static rings is shown to be $\Omega (gn)$.
This  theorem clearly holds also in dynamic rings.

\begin{theorem}
	\label{lower}
	A lower bound on the total number of agent moves required to solve the $g$-partial gathering problem 
	in dynamic rings is  $\Omega (gn)$ if $g\ge 2$. 
\end{theorem}

On the time complexity, the following theorem holds.
\if()
Intuitively, this is because
there exist an initial configuration and link-missings such that 
the distance between some agent $a_i$ and its nearest agent is $\Omega(n)$, 
which requires $\Omega(n)$ rounds for $a_i$ to meet with other agents.
\fi

\begin{theorem}
	\label{theo:time}
	A lower bound on the time complexity required to solve the $g$-partial gathering problem 
	in dynamic  rings is  $\Omega (n)$. 
\end{theorem}

\begin{proof}
	We assume that 
	$k < o(n)$ holds,  agent $a_0$ is placed at node $v_0$ and agents $a_1, a_2, \ldots , a_{k-1}$ are respectively placed at nodes 
	$v_{(n-1)-(k-1)+1}$, $v_{(n-1)-(k-1)+2}$, $\ldots$ , $v_{n-1}$ in the initial configuration, 
	and link $e_{n-1}$ (link connecting $v_{n-1}$ and $v_0$) is missing throughout the execution.
	Then, the nearest agent of $a_0$ is $a_1$ and the distance is $\Omega (n)$.
	Hence, $a_0$ requires at least $\Omega (n)$  rounds to meet with other agents, and the theorem follows.	
\end{proof}

\section{The case of $k \le 2g$}
\label{sec:2g}

In this section, we show the impossibility result in the case of $k \le 2g$.

\begin{theorem}
	\label{theo:2g}
	When $k \le 2g$ holds, 
	the $g$-partial gathering problem cannot be solved in dynamic rings. 
\end{theorem}

\begin{proof}
	When $k<2g$ holds, 
	the $g$-partial gathering problem is equivalent to the total gathering problem,
	and it is shown that total gathering requiring all the $k$ agents to meet at the same node is impossible in 1-interval connected rings \cite{gatheringDynamicRing}.
	Hence, in the following we show the impossibility in the case $k=2g$.
	We show this by contradiction, that is, we assume that 
	there exists an algorithm $\cal{A}$ to solve the $g$-partial gathering problem with $k = 2g$ agents in dynamic rings.
	Let $c_t$ be the first  configuration from which no agent moves anymore, that is, 
	agents termite the algorithm execution in the final configuration $c_{t'}$ for some $t'\ge t$ and 
	the placement of agents in $c_{t''}$ for any $t''$~($t''\ge t$) achieves $g$-partial gathering. 
	In addition, let an \textit{agent node} be a node where an agent exists.
	In the following, we consider the cases where the number of agent nodes in $c_{t-1}$ is 
	(i) two and (ii) at least three and 
	we show that agents cannot achieve $g$-partial gathering from any case.  
	
	First, we consider the case that (i) there exist exactly two agent nodes in $c_{t-1}$.
	Let $v_a$ and $v_b$ be the agent nodes and 
	$k_a$ (resp., $k_b$) be the number of agents at $v_a$ (resp., $v_b$).
	Without loss of generality, in $c_{t-1}$, 
	we assume that $k_a>k_b$ holds and 
	there exist $g+c$ agents and $g-c$ agents at $v_a$ and $v_b$, respectively,  for some positive integer $c\,(1\le c\le g-1$).
	Then, since $k = 2g$ holds, 
	the $g$-partial gathering-achieved placement of agents in $c_t$ is either that 
	(a) all the $k$ agents gather at a single node or 
	(b) some $g$ agents gather at some node and the other $g~(=k-g)$ agents gather at another node. 
	To reach configuration (a) or (b) from $c_{t-1}$, 
	obviously at least $c$ agents in total need to move through a link between $v_a$ and $v_b$.
	However, the adversary can make the link between $v_a$ and $v_b$ be missing in $c_{t-1}$, which means that 
	$g$-partial gathering is not achieved in $c_t$.
	Hence, agents cannot achieve $g$-partial gathering in case (i).
	
	\if()
	in order achieve $g$-partial gathering in $c_t$, 
	it is necessary that (a) all agent gather at a single node or 
	(b) some $g$ agents gather at some node and the other $g$ agents gather at another node. 
	the number of agent nodes in $c_t$ is one or two.
	Then, in order to achieve $g$-partial gathering, 
	in $c_{t-1}$, obviously at least $c$ agents in total need to move through the link between $v_a$ and $v_b$.
	However, the adversary can make a link between $v_a$ and $v_b$ be missing in $c_{t-1}$, which means that 
	$g$-partial gathering is not achieved in $c_t$.
	Hence, agents cannot achieve $g$-partial gathering in case (i).
	\fi
	
	Next, we consider the case that (ii) there exist at least three agent nodes in $c_{t-1}$.
	Since $k = 2g$ holds, there exist at least two agent nodes $v_a^\textit{less}$ and $v_b^\textit{less}$ with less than $g$ agents each in $c_{t-1}$.
	Let $e_a^f$ (resp., $e_a^b$) be the link connecting $v_a^\textit{less}$ and $v_a^\textit{less}$'s forward
	(resp., backward) neighboring node. 
	We define  $e_b^f$ and $e_b^b$ similarly. 
	Then, from $c_{t-1}$, to achieve $g$-partial gathering (i.e., reach configuration (a) or (b) in the previous paragraph), 
	it is obviously necessary that at least two links of 
	either $e_a^f$ or $e_a^b$, and either $e_b^f$ or $e_b^b$, are passed by an agent.
	However, the adversary can make $e_a^f$, $e_a^b$, $e_b^f$, or $e_b^b$ be missing in $c_{t-1}$, which means 
	agents cannot reach configuration (a) or (b) in $c_t$ and they cannot achieve $g$-partial gathering also from in case (ii).
	\if()
	that some agents at $v_a^\textit{less}$ leave the current node or some agents visit $v_a^\textit{less}$.
	Such a behavior is necessary also for $v_b^\textit{less}$ in $c_{t-1}$.
	
	Then, 	
	
	Then, letting $e_a^f$ and $e_a^b$ (resp., $e_b^f$ and $e_b^b$) be the links 
	connecting $v_a^\textit{less}$ to 
	
	However, 
	the adversary can make some link to achieve $g$-partial gathering to be missing, 
	and agents cannot solve the problem from $c_{t-1}$ also in case (ii).
	\fi
	Therefore, the theorem follows. 
	\if()
	and at least a part of agents at $v_a^\textit{less}$ and a part of agents at $v_b^\textit{less}$ need to gather at a single node in $c_{t-1}$.
	However, by the same discussion as that in case (i),
	the adversary can make a link between $v_a^\textit{less}$ and $v_b^\textit{less}$ be missing, which means
	$g$-partial gathering is not achieved in $c_t$ similarly to case (i).
	Thus, agents cannot achieve $g$-partial gathering also from case (ii).
	Thus, the theorem follows.
	\fi
\end{proof}

\section{The case of $2g+1\le k\le 3g-2$}
\label{sec:more2g}

In this section, when $2g+1\le k\le 3g-2$,
we proposed an algorithm to solve the $g$-partial gathering problem in dynamic rings 
with $O(n\log g)$ rounds and the total number of $O(gn\log g)$ moves. 
In this algorithm, all agents try to travel once around the ring to get IDs of all agents, and then
determine a single common node where all agents should gather.
However, it is possible that some agent cannot travel once around the ring and get IDs of all agents due to missing links.
Agents treat this by additional behaviors explained by the following subsections.
The algorithm comprises two phases: the selection phase and the gathering phase. 
In the selection phase, agents move in the ring and determine the \textit{gathering node} where they should gather. 
In the gathering phase, agents try to stay  at the gathering node.

\subsection{Selection phase}
\label{sec:more2gSelection}
The aim of this phase is that each agent achieves either of the following two goals: 
(i) It travels once around the ring and gets IDs of all agents, or 
(ii) it detects that all agents stay at the same node. 
To this end, 
we use an idea similar to \cite{gatheringDynamicRing} which 
considers total gathering in dynamic rings.
First, each agent $a_i$ writes its ID on the current whiteboard and then tries to move forward for $3n$ rounds.
During the movement, $a_i$ memorizes values of observed IDs to array $a_i.\textit{ids}$[\,].
After the $3n$ rounds, the number $a_i.\textit{nVisited}$ of nodes that $a_i$ has visited is 
(a) at least $n$ or (b) less than $n$ due to missing links.
In case (a), $a_i$ must have completed traveling once around the ring.
Thus, $a_i$ can get IDs of all $k$ agents (goal (i) is achieved).
Then, $a_i$ (and the other agents) select the gathering node $v_\textit{gather}$ as the node where the minimum ID \textit{min} is written.

In case (b) (i.e., $a_i$ has visited less than $n$ nodes during the $3n$ rounds), 
we show in Lemma \ref{lem:more2gSelection} that 
all $k$ agents stay at the same node
(goal (ii) is achieved).
This situation means that agents already achieve $g$-partial (or total) gathering, and 
they terminate the algorithm execution. 

\begin{table*}[t!]
	\centering 
	\setlength{\myheight}{0.8mm}
	\caption{Global variables used in proposed algorithms.}
	\label{table:variables}
	\footnotesize
	\begin{tabular}{|l|l|l|l|}
		
		\multicolumn{4}{l}{\textbf{Variables for  agent $a_i$}} \\
		\hline 
		Type &  Name & Meaning & Initial value \\ 
		\hline
		int  & $a_i.\textit{rounds}$ & 
		\begin{tabular}{l}
			number of rounds from some round 
		\end{tabular}
		& 1 \\ 
		
		int  & $a_i.\textit{nIDs}$   & 
		\begin{tabular}{l}
			number of different IDs that $a_i$ has observed from some round 
		\end{tabular}
		& 0 \\
		
		int  & $a_i.\textit{nVisited}$   &  
		\begin{tabular}{l}
			number of nodes that $a_i$ has ever visited 
		\end{tabular}
		& 0 \\
		
		int  & $a_i.\textit{rank}$   & 
		\begin{tabular}{l}
			ordinal number of how its ID is small \\
			among IDs of agents at the same node
		\end{tabular}
		&  0 \\

		int  & $a_i.\textit{dir}$   & 
		\begin{tabular}{l}
			direction to which $a_i$ tries to move (1: forward, -1: backward) 
		\end{tabular}		& 0 \\
		
		array & $a_i.\textit{ids}[]$ & 
		\begin{tabular}{l}
			sequence of IDs that $a_i$ has observed 
		\end{tabular}
		& $\perp$ \\
		
		\hline
		
		\multicolumn{4}{l}{\textbf{Variables for  node $v_j$}} \\
		\hline 
		Type &  Name & Meaning & Initial value \\ 
		\hline
		
		\hline 
		int & $v_j.{id}$ & 
		\begin{tabular}{l}
			ID stored by $v_j$ 
		\end{tabular}
		& $\perp$ \\
		
		int & $v_j.\textit{nAgents}$ &
		\begin{tabular}{l}
			number of agents staying at $v_j$ 
		\end{tabular}
		& 0 \\ 
		
		int  & $v_j.\textit{dir}$   & 
		\begin{tabular}{l}
			direction to which an agent group that visited $v_j$ \\
			for the first time tries to move (1: forward, -1: backward) \\
		\end{tabular}
		& 0 \\
		
		boolean & $v_j.{\textit{waiting}}$ & 
		\begin{tabular}{l}
			whether some agents that keep staying at $v_j$ exist  or not  
		\end{tabular}
		& false \\
		
		boolean & $v_j.{\textit{fMarked}}$ & 
		\begin{tabular}{l}
			whether $v_j$ is visited by a forward group  or not 
		\end{tabular}
		& false \\
		
		boolean & $v_j.{\textit{bMarked}}$ & 
		\begin{tabular}{l}
			whether $v_j$ is visited by a backward group  or not 
		\end{tabular}
		& false \\
		
		boolean & $v_j.{\textit{candi}}$ & 
		\begin{tabular}{l}
			whether $v_j$ is a gathering-candidate node  or not 
		\end{tabular}
		& false \\
		
		\hline 
		
	\end{tabular}\\
	
	\footnotesize
\end{table*}

The pseudocode of the selection phase is described in Algorithm \ref{algo:more2gSelection}.
Global variables are summarized in Table~\ref{table:variables}
(several variables are used in other sections).
Note that, during the selection phase, 
an agent $a_i$ may visit more than $n$ nodes and observe more than $k$ IDs when it  is blocked less than $2n$ times. 
In this case, $a_i$ memorizes only the first $k$ IDs to memorize the location of the gathering node correctly (line 8).
In addition, even if $a_i$ has visited at least $n$ nodes and got IDs of all $k$ agents, 
it is possible that another agent $a_j$ has blocked at least $2n+1$ times and it does not get IDs of all $k$ agents.
In this case, all agents stay at the same node $v_j$ (Lemma \ref{lem:more2gSelection}), and $a_i$ detects the fact 
when the number $v_j.\textit{nAgents}$ at the current node $v_j$ is equal to 
the number $a_i.\textit{nIDs}$ of different IDs that $a_i$ has ever observed (these values are equal to $k$).
Then, $a_i$ (and the other agents) terminate the algorithm execution since $g$-partial (or total) gathering is already achieved (line 14).

\begin{algorithm}[t!]
	\caption{The behavior of agent $a_i$ in the selection phase ($v_j$ is the current node of $a_i$.)} 
	\label{algo:more2gSelection}     
	
	\Main{}
	$v_j.\textit{id}:=a_i.\textit{id}$, 	$a_i.\textit{ids}[a_i.\textit{nIDs}] := v_j.\textit{id}$\;
	$a_i.\textit{nIDs} := a_i.\textit{nIDs}+1$, $v_j.\textit{nAgents} := v_j.\textit{nAgents} + 1$\;
	\While {$a_i.\textit{rounds} \le 3n$}
	{
		$v_j.\textit{nAgents} := v_j.\textit{nAgents} - 1$\;
		Try to move from the  current node  $v_j$ to the forward node $v_{j+1}$\;
		\If{$a_i$ reached $v_{j+1}$ (that becomes new $v_j$)}
		{
			$a_i.\textit{nVisited} := a_i.\textit{nVisited} + 1 $\;
			\If{$(v_j.\textit{id} \neq \perp) \land (a_i.\textit{nIDs} < k)$}
			{
				$a_i.\textit{ids}[a_i.\textit{nIDs}] := v_j.\textit{id}$, $a_i.\textit{nIDs} := a_i.\textit{nIDs}+1$\;
				
			}
			
		}
		$v_j.\textit{nAgents} := v_j.\textit{nAgents} + 1$, $a_i.\textit{rounds} := a_i.\textit{rounds}+1$\;
		
	}
	\textbf{if} $a_i.\textit{nVisited} < n$ \textbf{then} terminate the algorithm execution // all $k$ agents stay at the current node\;
	
	\If{$a_i.\textit{nVisited} \ge n$}
	{
		// $a_i$ traveled once around the ring and got IDs of all $k$ agents\;
		\textbf{if} $v_j.\textit{nAgents}=a_i.\textit{nIDs}$ \textbf{then} terminate the algorithm execution // all $k$ agents stay at the current node\;
		Let \textit{min} be the minimum ID among $a_i.\textit{ids}[\,]$ and select  the gathering node $v_\textit{gather}$ as a node where \textit{min} is written\; 
		Terminate the selection phase and enter the gathering phase \;
	}
\end{algorithm}

Concerning the selection phase, we have the following lemma. 

\begin{lemma}
	\label{lem:more2gSelection}
	After finishing Algorithm \ref{algo:more2gSelection}, 
	each agent achieves 
	either of the following two goals:
	(i) It travels once around the ring and gets IDs of all agents, or 
	(ii) it detects that all agents stay at the same node. 
	
\end{lemma}

\begin{proof}
	We use a proof idea similar to that in \cite{gatheringDynamicRing}, which 
	considers total gathering in dynamic rings and shows  
	that agents achieve goal (i) or (ii) after the movement.
	We consider the cases that the value of $a_i.\textit{nVisited}$ after executing Algorithm \ref{algo:more2gSelection} is
	(a) at least $n$ (line 12), and  
	(b) less than $n$ (line 11) in this order.
	First, if $a_i.\textit{nVisited}\ge n$ holds,
	clearly $a_i$ travels at least once around the ring and 
	gets IDs of all the $k$ agents.
	Thus, $a_i$ achieves goal (i).
	Next, we consider the case that (b) $a_i.\textit{nVisited}<n$ holds.
	In this case, we show that all agents stay at the same node.
	We show this by contradiction, that is, 
	we assume that some agents $a_i$ and $a_j$ exist at different nodes after executing Algorithm \ref{algo:more2gSelection}.
	We consider the distance $\textit{d}_\textit{ji}$ from $a_j$ to $a_i$, that is, 
	when $a_i$ is blocked and $a_j$ moves forward at some round, 
	the value of $d_\textit{ji}$  decreases by one and vice versa.
	Then, if $a_i$ is not blocked at some round,
	the value of $d_\textit{ji}$ increases by at most one since 
	all agents try to move forward.
	This happens at most $n-1$ rounds because of $a_i.\textit{nVisited} <n$.
	On the other hand, if $a_i$ is blocked at some round, 
	$d_\textit{ji}$ decreases by one (or is already 0).
	This happens at least $2n+1$ rounds.
	Then, since the value of $d_\textit{ji}$ is at most $n-1$ in the initial configuration,
	the value of $d_\textit{ji}$ after executing Algorithm \ref{algo:more2gSelection} is
	$\max\{(n-1)+(n-1)-(2n+1),0\} = 0$. 
	This means that $a_i$ and $a_j$ (and the other agents) stay at the same node, which is a contradiction. 
	Thus, agents achieve goal (ii) and 
	the lemma follows.	
\end{proof}

\subsection{Gathering phase}
\label{sec:more2gGathering}

In this phase, 
agents aim to achieve $g$-partial (or total) gathering by trying to visit the gathering node $v_\textit{gather}$. 
Concretely, for $3n$ rounds from the beginning of this phase, 
each agent $a_i$ tries to move forward until it reaches $v_\textit{gather}$.
If agents are blocked few times, all agents can reach $v_\textit{gather}$ and 
they achieve $g$-partial (or total) gathering.
However, it is possible that some agent cannot reach $v_\textit{gather}$ due to link-missings, resulting in a situation such that there exists a node with only less than $g$ agents. 
To treat this, we introduce a technique called \textit{splitting}.
Intuitively, in this technique, 
when at least $g+c~(c\ge 2)$ agents exist at some node, from there
some $c/2$ agents try to move forward and other $c/2$ agents try to move backward 
in order to visit the node with less than $g$ agents, and the remaining $g$ agents stay at the current node.
By repeating this behavior and additional behaviors explained later, 
agents aim to reduce the gap between $g$ and the number of agents at the node with less than $g$ agents,
and eventually achieve $g$-partial gathering. 

\begin{figure*}[t!]
	\centering
	\includegraphics[scale=0.7]{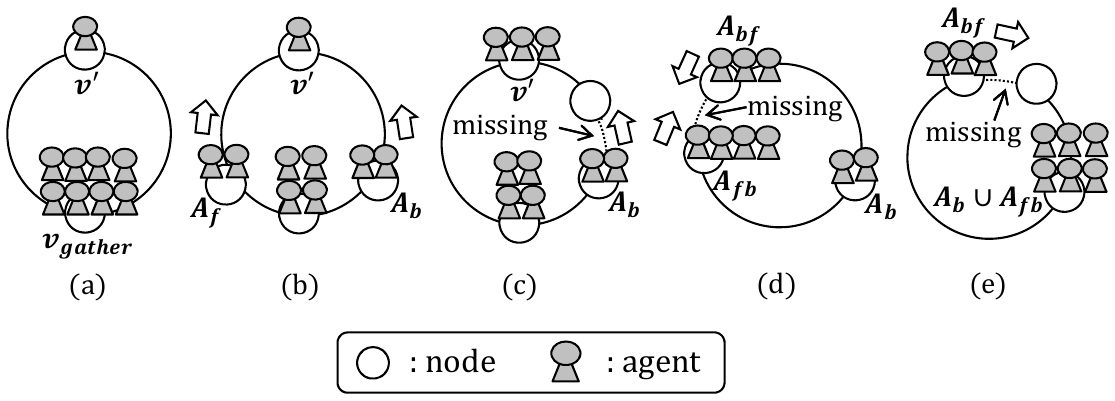}
	\caption{An execution example of the gathering phase when $2g+1\le k\le3g-2$~($g=4$).}
	\label{fig:2g+1}
\end{figure*} 

Concretely, after the $3n$ round from when agents tried to move forward to reach $v_\textit{gather}$, 
by the similar discussion of the proof of Lemma \ref{lem:more2gSelection}, 
all agents that do not reach $v_\textit{gather}$ stay at the same node.
Letting $v'$ be the node,
there are at most two nodes $v_\textit{gather}$ and $v'$ where agents exist after the movement.
An example in the case of $g = 4$ is given in Fig.~\ref{fig:2g+1}(a) (we omit nodes unrelated to the explanation).
Then, when agents recognize by knowledge of $k$ and $g$ that at least $g$ agents exist at both $v_\textit{gather}$ and $v'$,
they already achieve $g$-partial gathering and thus terminate the algorithm execution.
Otherwise, i.e., when $v_\textit{gather}$ or $v'$ has only less than $g$ agents, 
agents use the splitting technique and some agents try to visit the node with less than $g$ agents. 
This technique comprises at most $\lceil \log g\rceil$ subphases.
Without loss of generality, at the beginning of the splitting, 
we assume that there exist $g-c'~(1\le c'\le g-1)$ agents at $v'$. 
Then, letting $c = k-2g+c'$, there exist $k -(g-c') = g+c$ agents at $v_\textit{gather}$.
Notice that $c> c'$ and $c\ge 2$ hold.
First, each agent $a_i$ at $v_\textit{gather}$ calculates how small its ID is among the $g +c$ agents. 
We denote the ordinal number by $a_i.\textit{rank}$.
Then, if $1\le a_i.\textit{rank}\le g$, $a_i$ stays at $v_\textit{gather}$.
If $g+1\le a_i.\textit{rank}\le \lfloor c/2\rfloor$, 
it belongs to the \textit{forward agent group} $A_f$. 
Otherwise, i.e., if $\lfloor c/2\rfloor +1\le a_i.\textit{rank}\le g+c$, 
it belongs to the \textit{backward agent group} $A_b$. 
Thereafter, for $n$ rounds, 
the forward agent group $A_f$ (resp., the backward group $A_b$) tries to move forward (resp., backward)
until visiting $v'$ (Fig.~\ref{fig:2g+1}(b)). 
Then, since at most one link is missing at each round, $A_f$ or $A_b$ can move to the next node during the movement
and either of them can reach $v'$ (Fig.~\ref{fig:2g+1}(c)).

After the $n$ rounds, since $A_f$ or $A_b$ can reach $v'$,
the number of agents at $v'$ increases by at least $\lfloor c/2\rfloor$.
Without loss of generality, we assume that $A_f$ could reach $v'$ but 
$A_b$ had been blocked and could not reach $v'$.
Thereafter, all agents at $v_\textit{gather}$ and all agents at $v'$ try to move in the ring and 
visit the node where $A_b$ is currently staying.
Concretely, 
let $A_\textit{fb}$ be the agent group comprising all agents that stayed at $v_\textit{gather}$ for the previous $n$ rounds 
and $A_\textit{bf}$ be the agent group comprising 
all agents in $A_f$ and all agents that stayed at $v'$ for the previous $n$ rounds.
Then, for $n$ rounds, 
$A_\textit{fb}$ (resp., $A_\textit{bf}$) tries to move forward (resp., backward) until visiting the node where $A_b$ is staying (Fig.~\ref{fig:2g+1}(d)).
After that, if $A_\textit{fb}$ (resp., $A_\textit{bf}$) could not reach the node with $A_b$,
for $n$ rounds, it switches the direction and tries to move backward (resp., forward) until visiting the node with $A_b$. 
By there behaviors, if we assume that there does not exist the node with $A_b$, 
every node is visited by $A_\textit{fb}$ or $A_\textit{bf}$ for this $2n$ rounds.
Thus, $A_\textit{fb}$ or $A_\textit{bf}$ can visit the node with $A_b$ (Fig.~\ref{fig:2g+1}(e)). 
If $A_\textit{bf}$ could reach the node, there exist $g$ agents at the node with $A_\textit{fb}$ and 
$g-c'+c~(\ge g)$ agents at the node with $A_\textit{bf}$~(and $A_b$).
Hence, they achieve $g$-partial gathering and terminate the algorithm execution.
Otherwise, i.e., if $A_\textit{fb}$ could reach the node with $A_b$,
there exist $g + \lfloor c/2\rfloor$ agents at the node with $A_\textit{fb}$~(and $A_b$) and 
$g -c' + \lceil c/2\rceil $ agents at the node with $A_\textit{bf}$.
Then, the gap between $g$ and the number of agents at the node with less than $g$ agents halves. 
Thus, by executing such a subphase at most $\lceil \log g\rceil$ times,
agents can make a configuration such that at least $g$ agents exist at every node where an agent exists
($g$-partial gathering is achieved).

\begin{algorithm}[t!]
	\caption{The behavior of agent $a_i$ in the gathering phase 
		($v_j$ is the current node of $a_i$.)}
	\label{algo:more2gGathering}       
	
	\Main{}
	$a_i.\textit{rounds} := 1$\;
	
	\While {$a_i.\textit{rounds} \le 3n$}
	{
		
		\If {$v_j \neq v_\textit{gather}$}
		{
			$v_j.\textit{nAgents} := v_j.\textit{nAgents} - 1$\;
			Try to move from the current node $v_j$ to the forward node $v_{j+1}$ (that may become new $v_j$ depending on the existence of a missing link)\;
			$v_j.\textit{nAgents} := v_j.\textit{nAgents} + 1$\;
		}
		$a_i.\textit{rounds} := a_i.\textit{rounds}+1$\;
	}
	\textbf{if} $(v_j.\textit{nAgents}\ge g) \land (k-v_j.\textit{nAgents}\ge g)$ \textbf{then}
	terminate the algorithm execution // at least $g$ agents already exist at both nodes with agents\;
	\textbf{else if } $v_j.\textit{nAgents} \ge g+2$ \textbf{then} \textit{More}()\;
	\textbf{else} \textit{Less}()\;	
	
\end{algorithm}

\begin{algorithm}[t!]
	\caption{Procedure \textit{More}() ($v_j$ is the current node of $a_i$.)}
	\label{algo:more2gMore}       
	
	\Main{}
	Calculate $a_i.\textit{rank}$\;
	
	\textbf{if} $1\le a_i.\textit{rank}\le g$ \textbf{then} $a_i.\textit{dir} := 0 $\;
	\textbf{else if} $g+1\le a_i.\textit{rank}\le \lfloor (v_j.\textit{nAgents}-g)/2\rfloor$ \textbf{then} $a_i.\textit{dir} := 1 $\;
	\textbf{else} $a_i.\textit{dir} := -1 $\;
	
	Moving($a_i.\textit{dir}$)\;
	
	\textbf{if} $v_j.\textit{nAgents} <g$ \textbf{then} $a_i.\textit{dir} :=0, v_j.\textit{waiting} := \textit{true}$\;
	\textbf{else if} $v_j.\textit{waiting} =\textit{false}$ \textbf{then} $a_i.\textit{dir} :=1$\;
	\textbf{else} $a_i.\textit{dir} :=-1, v_j.\textit{waiting} = \textit{false}$\;
	LatterMove($a_i.\textit{dir}$)\;
	
\end{algorithm}

\begin{algorithm}[t!]
	\caption{Procedure \textit{Less}()	($v_j$ is the current node of $a_i$.)}
	\label{algo:more2gLess}

	\Main{}
	$a_i.\textit{rounds}:=1, v_j.\textit{waiting} := \textit{true}$\;
	\While{$a_i.\textit{rounds}\le n$}
	{
		Stay at the current node $v_j$\;
		$a_i.\textit{rounds} := a_i.\textit{rounds} +1$\;
	}
	$a_i.\textit{dir}:=-1, v_j.\textit{waiting} :=\textit{false}$\;
	LatterMove($a_i.\textit{dir}$)
	
	\if()
	Moving($a_i.\textit{dir}$)\;
	$a_i.\textit{dir}:= a_i.\textit{dir} \times (-1)$\;
	Moving($a_i.\textit{dir}$)\;
	\fi

\end{algorithm}

\begin{algorithm}[t!]
	\caption{Procedure \textit{Moving}($a_i.\textit{dir}$) ($v_j$ is the current node of $a_i$.)}
	\label{algo:more2gMove}       
	
	\Main{}
	$a_i.\textit{rounds} := 1$\;	
	\While {$a_i.\textit{rounds}\le n$}
	{
		\If {$a_i.\textit{dir} := 0$}
		{stay at the current node $v_j$\;}
		\Else
		{	
			\While{$v_j.\textit{waiting}=\textit{false}$}
			{
				$v_j.\textit{nAgents} := v_j.\textit{nAgents}-1$\;
				Try to move from the current node $v_j$ to node $v_{j+a_i.\textit{dir}}$ (that may become new
				$v_j$ depending on the existence of a missing link)\;	
				$v_j.\textit{nAgents} := v_j.\textit{nAgents}+1$
			}
		}
		$a_i.\textit{round}:=a_i.\textit{round}+1$\;	
	}
	
\end{algorithm}

\begin{algorithm}[t!]
	\caption{Procedure \textit{LatterMove}($a_i.\textit{dir}$) ($v_j$ is the current node of $a_i$.)}
	\label{algo:more2gMove2}       
	
	Moving($a_i.\textit{dir}$)\;
	$a_i.\textit{dir}:= a_i.\textit{dir} \times (-1)$\;
	Moving($a_i.\textit{dir}$)\;
	
	\textbf{if} $(v_j.\textit{nAgents}\ge g) \land (k-v_j.\textit{nAgents}\ge g)$ \textbf{then}
	terminate the algorithm execution\;
	\textbf{else if } $v_j.\textit{nAgents} \ge g+2$ \textbf{then} \textit{More}()\;
	\textbf{else} \textit{Less}()\;	
\end{algorithm}

The pseudocode of the gathering phase is described in Algorithm \ref{algo:more2gGathering}.
After executing the first $3n$ rounds of this phase, 
each agent $a_i$ counts the number of agents at the current node through $v_j.\textit{nAgents}$.
If $v_j.\textit{nAgents}\ge g+2$ and the number of agents at the other node with an agent is less than $g$, 
$a_i$ executes procedure \textit{More}() to decide whether it belongs to $A_f$ or  $A_b$, or 
keeps staying at $v_j$, depending on $a_i.\textit{rank}$.
The pseudocode of \textit{More}() is described in Algorithm \ref{algo:more2gMore}.
If $v_j.\textit{nAgents}<g$, $a_i$ executes procedure \textit{Less}() to keep staying at $v_j$ for $n$ rounds 
and then resume moving.
The pseudocode of \textit{Less}() is described in Algorithm \ref{algo:more2gLess}. 
In Algorithms \ref{algo:more2gMore} and \ref{algo:more2gLess},
we omit how $a_i$ calculates $a_i.\textit{rank}$ for simplicity, and 
agents use procedures \textit{Moving}() and 
\textit{LatterMove}() to move in the ring, whose pseudocodes are described in 
Algorithms \ref{algo:more2gMove} and \ref{algo:more2gMove2}, respectively.

Concerning the gathering phase, we have the following lemma. 

\begin{lemma}
	\label{lem:more2gGathering}
	After finishing Algorithm \ref{algo:more2gGathering}, agents achieve $g$-partial gathering.
\end{lemma}

\begin{proof}
	First, for $3n$ rounds from the beginning of the gathering phase,
	each agent tries to move forward to reach $v_\textit{gather}$.
	Then, by the similar discussion of the proof of Lemma \ref{lem:more2gSelection},
	all agents that do not reach $v_\textit{gather}$ stay at the same node $v'$.
	When at least $g$ agents exist at both $v_\textit{gather}$ and $v'$,
	they already achieve $g$-partial gathering and can terminate the algorithm execution there by the knowledge of $g$ and $k$.
	On the other hand, when only less than $g$ agents exist at $v'$ (resp., $v_\textit{gather}$),
	at least $g+2$ agents exist at $v_\textit{gather}$ (resp., $v'$) since we consider the case of $2g+1\le k \le 3g-2$.
	Without loss of generality, we assume that there exist $g-c'~(1\le c'\le g-1)$ agents at $v'$.
	Then, letting $c=k-2g+c'$, there exist $k-(g-c') = k+c$ agents at $v_\textit{gather}$.
	Notice that $c>c'$ and $c\ge 2$ hold.
	
	From such a situation, 
	agents at $v_\textit{gather}$ execute \textit{More}() and then execute \textit{Moving}().
	In Procedure \textit{More}(), 
	$\lfloor c/2\rfloor$ agents belong to a forward group $A_f$ and try to move forward and $\lceil c/2\rceil$ agents 
	belong to a backward group $A_b$ and try to move backward in order to visit $v'$ (lines 3 -- 5).
	Since at most one link is missing, $A_f$ or $A_b$ can move to the next node unless they already visit $v'$,
	and hence either of them can visit $v'$ within $n$ rounds. 
	Without loss of generality, we assume that $A_f$ could reach $v'$ but $A_b$ could not do it due to link-missings.
	Thereafter, in order to visit the node with $A_b$, by Procedure \textit{LatterMove}(),
	agents that stayed at $v_\textit{gather}$ (resp., $v'$) belong to an agent group $A_\textit{fb}$ (resp., $A_{\textit{bf}}$) and try to move forward (resp., backward) for $n$ rounds and then 
	try to move backward (resp., forward) for $n$ rounds.
	Then, when we assume that $A_b$ does not exist, 
	$A_\textit{fb}$ and  $A_\textit{bf}$ can visit all nodes for this $2n$ rounds, 
	and thus either of them can reach the node with $A_b$.
	When  $A_\textit{bf}$ reaches the node, at least $g$ agents exist at both nodes with agents 
	and the can terminate the algorithm there. 
	When $A_\textit{fb}$ reaches the node, there exist $g+\lceil c/2\rceil$ agents at the node where $A_\textit{fb}$ exists,
	and there exist $g-c'+\lfloor c/2\rfloor$ agents exist at the node where $A_\textit{bf}$ exists~(Fig.~\ref{fig:2g+1}(e)). 
	Since $c>c'$ holds, the gap between $g$ and the number of agents at the node with less than $g$ agents halves 
	by the above behaviors.
	Thus, by repeating such behaviors at most $\lceil \log g\rceil$ times, 
	agents can achieve $g$-partial gathering. Therefore, the lemma follows. 	 
\end{proof}

We have the following theorem for the proposed algorithm.

\begin{theorem}
	\label{theo:2g+1}
	When $2g+1\le k\le 3g-2$ holds, 
	the proposed algorithm solves the $g$-partial gathering problem in dynamic rings with 
	$O(n\log g)$ rounds and  the total number of $O(gn\log g)$ moves. 
\end{theorem}

\begin{proof}
	By Lemmas \ref{lem:more2gSelection} and \ref{lem:more2gGathering},
	the proposed algorithm can solve the $g$-partial gathering problem. 
	In the following, we analyze the time complexity and the required total number of agent moves. 
	
	First, in the selection phase, each agent tries to move forward for $3n$ rounds in order to 
	determine the gathering node $v_\textit{gather}$.
	Since $k =O(g)$ holds in this section, 
	This requires $O(n)$ rounds and the total number of $O(gn)$ moves. 
	Next, in the gathering phase, for $3n$ rounds,
	each agent first tries to move forward to reach $v_\textit{gather}$,
	which requires $O(n)$ rounds and the total number of $O(gn)$ moves. 
	Thereafter, 
	(i) by Procedure \textit{More}(), for $n$ rounds, 
	several agents from the node with at least $g+2$ agents move forward or backward to visit the node with less than $g$ agents, and 
	(ii) by Procedure \textit{LatterMove}(), for $2n$ rounds, 
	agent groups $A_\textit{fb}$ and $A_\textit{bf}$ move forward or backward 
	to visit the node with less than $g$ agents. 
	The above behaviors require $O(n)$ rounds and the total number of $O(gn)$ moves. 
	Since agents repeat such behaviors at most $\lceil \log g\rceil$ times until achieving $g$-partial gathering,
	the total time complexity is $O(n\log g)$ and the total move complexity is $O(gn\log g)$.
	Therefore, the theorem follows. 
\end{proof}

\section{The case of $3g-1\le k\le 8g-4$}
\label{sec:more3g}
In this section, when $3g-1\le k\le 8g-4$, 
we propose a naive algorithm to solve the problem with $O(n)$ rounds and the total number of $O(kn)$ moves.
Since $k = O(g)$ holds in this case, this algorithm is asymptotically in terms of 
both the time and move complexities, similar to the third algorithm explained in Section \ref{sec:more8g}.
The algorithm tactics is the same as that in Section \ref{sec:more2g}, that is, 
all agents try to travel once around the ring to get IDs of all agents, and then
try to gather at a common single gathering node $v_\textit{gather}$. 
Hence, similarly to Section \ref{sec:more2g}, 
the algorithm comprises the selection phase and the gathering phase. 
The selection phase is exactly the same as that in Section \ref{sec:more2gSelection}, 
and we revise the gathering phase to reduce the total number of agent moves from $O(gn\log g)$ to $O(gn)$,
explained by the following paragraphs. 

After agents finish the selection phase, by Lemma \ref{lem:more2gSelection},
(i) all agents travelled once around the ring and get IDs of all agents 
(in this case all agents do not stay at the same node), or 
(ii) all agents already stay at the same node. 
In case (ii), agents already achieve $g$-partial gathering and terminate the algorithm execution there. 
In case (i), similarly to Section \ref{sec:more2gGathering},
each agent first tries to move forward for $3n$ rounds to visit the gathering node $v_\textit{gather}$.
Then, all agents that do not reach $v_\textit{gather}$ stay at the same node $v'$ similarly to Section \ref{sec:more2gGathering}.
However, different from Section \ref{sec:more2gGathering}, when 
there exist only less than $g$ agents at $v'$ (resp., $v_\textit{gather}$),
there exist at least $2g$ agents at $v_\textit{gather}$ (resp., $v'$)
since we consider the case of $3g-1\le k\le 8g-4$ in this section.
We call such a node with at least $2g$ agents $v_\textit{more}$. 
Notice that it is possible that at least $2g$ agents exist at both $v_\textit{gather}$ and $v'$. 
Then, we use this fact and modify (simplify) the splitting technique.
Intuitively, from $v_\textit{more}$, 
an agent group with at least $g$ agents tries to move forward and another agent group with at least $g$ agents tries to move backward.
In addition, when an agent group with at least $g$ agents visits a node where less than $g$ agents exist, 
the less than $g$ agents join the agent group 
and try to move in the same direction as that of the group.
By this behavior, it does not happen that all agents are blocked, and 
agents can eventually terminate in a configuration such that at least $g$ agents exist at each node where an agent exists. 

Concretely, 
let $k' (\ge 2g)$ be the number of agents staying at $v_\textit{more}$.
Then, each agent $a_i$ at $v_\textit{more}$ first calculates its $a_i.\textit{rank}$.
If $1\le a_i.\textit{rank}\le g$ holds, it belongs to a forward agent group $A_f$ and 
tries to move forward. 
Else if ($k' < 3g) \lor (g+1\le a_i.\textit{rank}\le 2g)$ holds, 
it belongs to the backward agent group $A_b$ and tries to move backward.
If $a_i$ does not satisfy any of the above conditions, 
it terminates the algorithm execution because 
there still exist at least $g$ agents even after $A_f$ and $A_b$ leave $v_\textit{more}$. 
While $A_f$ and $A_b$ move in the ring, if $A_f$ (resp., $A_b$) visits a new  node $v_j$,
it sets a flag $v_j.\textit{fMarked}$ (resp., $v_j.\textit{bMarked}$) representing that 
$v_j$ is visited by $A_f$ (resp., $A_b$).
These flags are used for an agent group $A$ to check whether or not 
the current node is visited by another agent group and $A$ can stop moving in the ring.
In addition, if $A_f$ (resp., $A_b$) visits a node with less than $g$ agents, 
the less than $g$ agents join $A_f$ (resp., $A_b$) and try to move forward (resp., backward). 
However, it is possible that 
the number \textit{num} of agents in the updated group is more than $2g$.
In this case, using their IDs, 
only $g$ agents continue to try moving and the remaining $\textit{num} - g$ agents terminate the algorithm execution at the current node.
By this behavior, 
each link is passed by at most $2g$ agents and the total number of moves for agent groups can be reduced to $O(gn)$
(this technique is used in Section \ref{sec:more8g}).
Moreover, since $A_f$ or $A_b$ can visit a next node at each round even when some link is missing, 
when $A_f$ (resp., $A_b$) repeats such a behavior for $n$ rounds or 
until it visits some node $v_j$ with $v_j.\textit{bMarked} = \textit{true}$ (resp., $v_j.\textit{fMarked} = \textit{true}$), 
$A_b$ and $A_f$ can visit all the nodes in the ring in total and $g$-partial gathering is achieved. 

An example is given in Fig.\,\ref{fig:more3gGathering}
(we omit nodes unrelated to the explanation).
From (a) to (b), a backward group $A_b$ visits a node with two $(<g)$ agents, and 
the two agents join $A_b$.
Then, since the number of agents in the updated $A_b$ is $7~(>2g)$,
only three agents continue to try moving and the remaining four agents terminate the algorithm execution there ((b) to (c)). 
From (c) to (d), we assume that 
a forward agent group $A_f$ continues to be blocked due to a missing link.
Even in this case, $A_b$ can continue to move since there is only one missing link at each round.
When $A_f$ (resp., $A_b$) visits a node with a flag set by $A_b$ (resp., $A_f$) like (e), or 
$n$ rounds passed from when agent groups started trying to move, 
agents achieve $g$-partial gathering.  

\begin{figure*}[t!]
	\centering
	\includegraphics[scale=0.7]{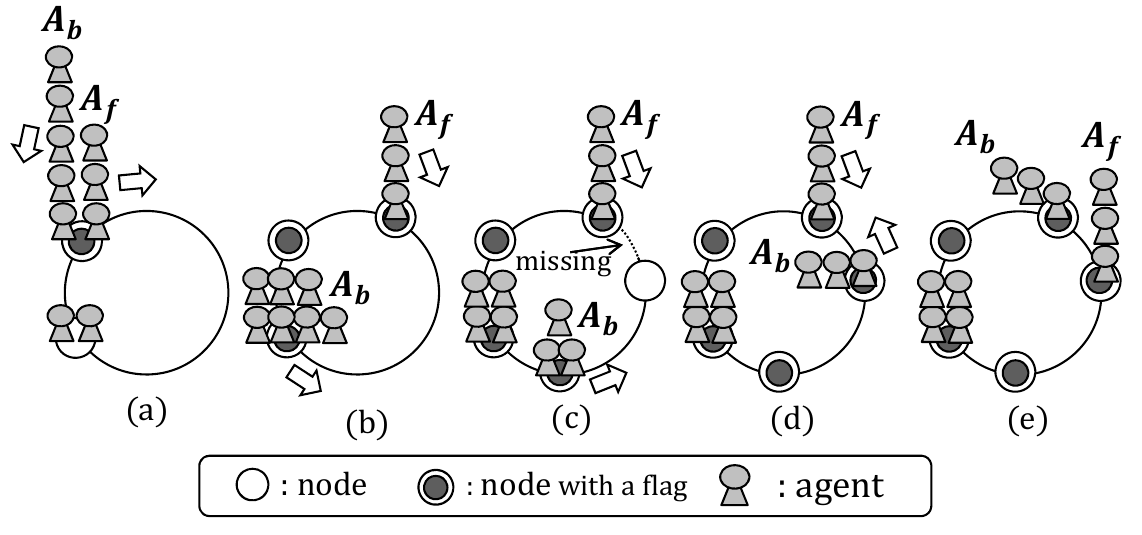}
	\caption{An execution example of the gathering phase when $3g-1\le k\le 8g-4$ ($g=3$).}
	\label{fig:more3gGathering}
\end{figure*} 

\begin{algorithm}[t!]
	\caption{The behavior of agent $a_i$ in the gathering phase 
		($v_j$ is the current node of $a_i$.)}
	\label{algo:more3gGathering}       
	
	\Main{}
	$a_i.\textit{rounds} := 1$\;
	
	\While {$a_i.\textit{rounds} \le 3n$}
	{
		
		\If {$v_j \neq v_\textit{gather}$}
		{
			$v_j.\textit{nAgents} := v_j.\textit{nAgents} - 1$\;
			Try to move from the current node $v_j$ to the forward node $v_{j+1}$ (that may become new $v_j$ depending on the existence of a missing link)\;
			$v_j.\textit{nAgents} := v_j.\textit{nAgents} + 1$\;
		}
		$a_i.\textit{rounds} := a_i.\textit{rounds}+1$\;
	}
	\textbf{if} $g\le v_j.\textit{nAgents} \le 2g-1$ \textbf{then}
	terminate the algorithm execution\;
	\textbf{else if } $v_j.\textit{nAgents} \ge 2g$ \textbf{then} \textit{More2}()\;
	\textbf{else} \textit{Less2}()\;

\end{algorithm}

The pseudocode of the gathering phase is described in Algorithm \ref{algo:more3gGathering}. 
After executing the first $3n$ rounds of this phase, 
each agent $a_i$ counts the number of agents at the current node through $v_j.\textit{nAgents}$.
If the number 
is at least $2g$ (i.e., $v_j$ is $v_\textit{more}$),
$a_i$ executes procedure \textit{More2}() and decides whether it belongs to $A_f$ or  $A_b$, or 
terminates the algorithm execution,  depending on $a_i.\textit{rank}$.
The pseudocode of \textit{More2}() is described in Algorithm \ref{algo:more}.
If the number of agents at $v_j$ is less than $g$, $a_i$ executes procedure \textit{Less2}() and keeps staying at $v_j$ until 
a forward or backward agent group visits $v_j$.
The pseudocode of \textit{Less2}() is described in Algorithm \ref{algo:less}. 
In Procedures \ref{algo:more} and \ref{algo:less}, 
agents use procedure \textit{Moving2}() to move in the ring, whose pseudocode is described in Algorithm \ref{algo:moving}.

\begin{algorithm}[t!]
	\caption{Procedure \textit{More2}() ($v_j$ is the current node of $a_i$.)}
	\label{algo:more}

	\Main{}
	$a_i.\textit{rounds} := 1$\;
	Calculate $a_i.\textit{rank}$\;
	
	\textbf{if} $1\le a_i.\textit{rank}\le g$ \textbf{then} $a_i.\textit{dir} := 1 $\;
	\textbf{else if} 
	$(v_j.\textit{nAgents} < 3g) \lor 
	(g+1 \le a_i.\textit{rank}\le 2g)$
	\textbf{then} $a_i.\textit{dir} := -1$\;
	\textbf{else} Terminate the algorithm execution // There still exist at least $g$ agents even after 
	the forward group and the backward group leave $v_j$\;
	\textit{Moving2}()\;
	
\end{algorithm}

\begin{algorithm}[t!]
	\caption{Procedure \textit{Less2}()	($v_j$ is the current node of $a_i$.)}
	\label{algo:less}

	\Main{}
	$a_i.\textit{rounds} := 1$\;
	$v_j.\textit{waiting} := \textit{true}$\;
	\While{true}
	{
		$a_i.\textit{rounds} := a_i.\textit{rounds} +1$\;
		\textbf{if} ($v_j.\textit{fMarked}= \textit{true}) \land (v_j.\textit{bMarked}= \textit{true})$
		\textbf{then} terminate the algorithm execution  // A forward group and a backward group visited $v_j$ simultaneously\;
		\If {($v_j.\textit{fMarked}= \textit{true}) \lor (v_j.\textit{bMarked}= \textit{true})$}
		{
			Calculate $a_i.\textit{rank}$ \;
			\textbf{if} ($v_j.\textit{nAgents} \ge 2g$) $\land$ ($a_i.rank\ge g+1$) \textbf{then} terminate the algorithm execution \;
			\textbf{else} $a_i.\textit{dir}:= v_j.\textit{dir}$ and execute \textit{Moving2}()\;
		}
	}
	
\end{algorithm}

In addition, it is possible that a forward group and a backward group visit some node simultaneously. 
This implies that all the nodes to be visited by the forward and backward groups are already visited,
and thus they terminate the algorithm execution there (line 5 in Algorithm \ref{algo:less} and line 12 in Algorithm \ref{algo:moving}).
Moreover, it is possible that some link continues to be missing for a long time 
when $A_f$ and $A_b$ are trying to move in the ring; 
$A_f$ and $A_b$ continue to be blocked by the same missing link $e'$ and 
$A_f$ (resp., $A_b$) cannot observe a node $v_j$ with $v_j.\textit{bMarked} = \textit{true}$
(resp., $v_j.\textit{fMarked} = \textit{true}$).
Even in this case, before they are blocked by $e'$,
$A_f$ or $A_b$ can visit a next node at each round and 
they can visit all the $n$ nodes in total within $n$ rounds.
Hence, even if agents do not observe a flag for termination, they can terminate the algorithm execution after the $n$ rounds from when 
they start trying to move 
(lines 1 and 17 in Algorithm \ref{algo:moving}).

Concerning the gathering phase, we have the following lemma.

\begin{algorithm}[t!]
	\caption{Procedure \textit{Moving2}()	($v_j$ is the current node of $a_i$.)}
	\label{algo:moving}       
	
	\Main{}
	
	\While{$a_i.\textit{rounds} \le n$}
	{
		$a_i.\textit{rounds} := a_i.\textit{rounds} +1$\;
		\textbf{if} ($a_i.\textit{dir} = 1) \land (v_j.\textit{fMarked} = \textit{false})$ \textbf{then} $v_j.\textit{fMarked} := \textit{true}$\;
		\textbf{if} ($a_i.\textit{dir} = -1) \land (v_j.\textit{bMarked} = \textit{false})$ \textbf{then} $v_j.\textit{bMarked} := \textit{true}$\;
		
		$v_j.\textit{nAgents} := v_j.\textit{nAgents} -1$\;
		Try to move from the current node $v_j$ to node $v_{j+a_i.\textit{dir}}$ (that may become new $v_j$ depending on the existence of a missing link)\;
		$v_j.\textit{nAgents} := v_j.\textit{nAgents} +1$\;
		
		\If{$a_i$ reached $v_{j+a_i.\textit{dir}}$ (that becomes new $v_j$)}
		{
			\textbf{if} 
			$((a_i.\textit{dir} = 1) \land (v_j.\textit{bMarked} = \textit{true})) \lor 
			(a_i.\textit{dir} = -1) \land (v_j.\textit{fMarked} = \textit{true})) $
			\textbf{then} 
			terminate the algorithm execution \;
			\textbf{if} ($a_i.\textit{dir} = 1)$ \textbf{then} $v_j.\textit{fMarked} := \textit{true}$\;
			\textbf{if} ($a_i.\textit{dir} = -1$) \textbf{then} $v_j.\textit{bMarked} := \textit{true}$\;
			\textbf{if} ($v_j.\textit{fMarked} = \textit{true}) \land (v_j.\textit{bMarked} = \textit{true})$
			\textbf{then} 
			terminate the algorithm execution // A forward group and a backward group visited $v_j$ simultaneously \;
			$v_j.\textit{dir}:= a_i.\textit{dir}$\;
			\If{$v_j.\textit{waiting} = \textit{true}$}		
			{
				Update $a_i.\textit{rank}$\;
				\textbf{if} ($v_j.\textit{nAgents} \ge 2g$) $\land$ ($a_i.rank\ge g+1$) \textbf{then} terminate the algorithm execution \;
			}
		}		
	}
	Terminate the algorithm execution \;
	
\end{algorithm}

\begin{lemma}
	\label{lem:more3gGathering}
	After finishing Algorithm \ref{algo:more3gGathering}, agents achieve $g$-partial gathering.
\end{lemma}

\begin{proof}
	First, for $3n$ rounds from the beginning of the gathering phase,
	each agent tries to move forward for reaching $v_\textit{gather}$.
	Then, by the similar discussion of the proof of Lemma \ref{lem:more2gSelection},
	all agents that do not reach $v_\textit{gather}$ stay at the same node $v'$.
	In addition, when only less than $g$ agents exist at $v'$ (resp., $v_\textit{gather}$),
	at least $2g$ agents exist at $v_\textit{gather}$ (resp., $v'$) since we consider the case of $k\ge 3g-1$.
	Thus, agents at $v_\textit{gather}$ or $v'$ execute \textit{More2}() and then execute \textit{Moving2}().
	In addition, when some forward group $A_f$ or backward group $A_b$ executing \textit{Moving2}() visits a node with less than $g$ agents, 
	the less than $g$ agents join the group.
	If the number \textit{num} of agents in the updated group is more than $2g$, by using their IDs, 
	only $g$ agents continue to try moving, and the remaining $\textit{num}-g\,(\ge g)$ agents terminate the algorithm execution there. 
	Thus, even after $A_f$ or $A_b$ leaves node $v_j$ where there originally exist less than $g$ agents, 
	it is guaranteed that (if any) $g$ agents exist at $v_j$.
	Moreover, since $A_f$ tries to move forward and 
	$A_b$ tries to move backward,
	$A_f$ or $A_b$ can visit a next node at each round. 
	Thus, within $n$ rounds, all nodes in the ring are visited by $A_f$ or $A_b$ and hence 
	there is no node where between 1 and $g-1$ agents exist.
	Thus, 	agents achieve $g$-partial gathering and the lemma follows. 
\end{proof}

We have the following theorem for the proposed algorithm.

\begin{theorem}
	\label{theo:smallK}
	When $3g-1\le k\le 8g-4$ holds, 
	the proposed algorithm solves the $g$-partial gathering problem in dynamic rings with 
	$O(n)$ rounds and 	the total number of $O(kn)$ moves. 
\end{theorem}

\begin{proof}
	By Lemmas \ref{lem:more2gSelection} and \ref{lem:more3gGathering},
	the proposed algorithm can solve the $g$-partial gathering problem. 
	In the following, we analyze the time complexity and the required total number of agent moves. 
	
	First, in the selection phase, each agent tries to move forward for $3n$ rounds in order to 
	determine the gathering node $v_\textit{gather}$.
	This requires $O(n)$ rounds and the total number of $O(kn)$ moves. 
	Next, in the gathering phase,  for $3n$ rounds, 
	each agent first tries to move forward to reach $v_\textit{gather}$,
	which requires $O(n)$ rounds and the total number of $O(kn)$ moves. 
	Thereafter, when at least $2g$ agents exist at some node, 
	from there a forward agent group $A_f$ and a backward agent group $A_b$ move in the ring to achieve $g$-partial gathering. 
	As discussed in the proof of Lemma \ref{lem:more3gGathering},
	a forward group or a backward group can visit a next node at each round, and thus 
	the movement for agent groups finishes within $n$ rounds. 
	In addition, since the number of agents of $A_f$ and that of $A_b$  are respectively at most $2g$ and 
	$A_f$ and $A_b$ visit $n, n+1,$ or $n+2$ nodes in total,
	each link is passed by at most $4g$ times during the movement for agent groups. 
	Thus, the total number of agent moves in the movement for the agent groups is $O(gn)$.
	Therefore, agents achieve $g$-partial gathering with $O(n)$ rounds and the total number of $O(kn)$ moves, and the theorem follows. 
\end{proof}

\section{The case of $k\ge 8g-3$}
\label{sec:more8g}
In this section, when $k\ge 8g-3$, 
we propose an algorithm to solve the problem with $O(n)$ rounds and the total number of $O(gn)$ (i.e., optimal) moves.
Since the move complexity  is not $O(kn)$ but $O(gn)$, 
it is not possible that all agents try to travel once around the ring as in Sections \ref{sec:more2g} and \ref{sec:more3g}.
Hence, in this section agents aim to reduce the total number of moves using distinct IDs and the fact of $k\ge 8g-3$.
The algorithm comprises three phases:
the semi-selection phase, the semi-gathering phase, and the achievement phase. 
In the semi-selection phase, 
agents select a set of \textit{gathering-candidate nodes} each of where 
at least $2g$ agents may gather.
In the semi-gathering phase, 
agents try to stay at a gathering-candidate node.
As a result,  at least $2g$ agents gather at some node 
(the node may not be a gathering-candidate node due to link-missings).
In the achievement phase, 
agents achieve $g$-partial gathering by the same method as that for the gathering phase in Section \ref{sec:more3g}.

\subsection{Semi-selection phase}
\label{sec:semiSelection}

The aim of this part is to select a set of gathering-candidate nodes each of where 
at least $2g$ agents may gather.
A possible approach is that 
each agent $a_i$ moves forward and backward for getting IDs of its 1-st, 2-nd, $\ldots , (2g-1)$-st forward agents
and IDs of its 1-st, 2-nd, $\ldots , (2g-1)$-st backward agents, and then 
returns to its initial node. 
Here, the $i$-th ($i\neq 0$) 
forward (resp., backward)
agent ${a'}$ of  agent $a$ represents the agent such that $i-1$ agents exist between $a$ and $a'$ 
in ${a}$'s forward (resp., backward)
direction in the initial configuration. 
Thereafter, $a_i$ compares its ID and the obtained $4g-2$ IDs.  
If its ID is the minimum, 
$a_i$ selects its initial node as a gathering-candidate node $v_\textit{candi}$. 
Then, the $2g-1$ agents existing in $a_i$'s backward direction try to move forward  to stay at $v_\textit{candi}$ and 
eventually $2g$ agents may gather at $v_\textit{gather}$.
However,  since we consider 1-interval connected rings, 
there are two problems: 
(1) it is possible that no gathering-candidate node is selected since
some agent may not be able to collect $4g-2$ IDs due to link-missings, and 
(2) even if a gathering-candidate node $v_\textit{candi}$ is selected, 
it is possible that some agent cannot reach $v_\textit{candi}$ due to link-missings and 
only less than $2g$ agents gather at each node. 

To treat  these problems, each agent $a_i$ in this phase keeps trying to move forward, 
tries  to observe more than $4g-2$ IDs, and considers some observed ID as its own ID when it observed the necessary number of IDs. 
Concretely, for $3n$ rounds, 
each agent $a_i$ tries to move forward until it observes $10g-4$ IDs or 
at least $2g$ agents exist at the current node. 
Thereafter,  $a_i$ determines its behavior depending on whether it observed at least  $8g-3$ IDs or not.
If $a_i$ did not observe at least $8g-3$ IDs, 
we show in Lemma \ref{lem:semi-selection} that 
at least $2g$ agents exist at some node $v_j$ and then
a flag $v_j.\textit{candi}$ is set to true to represent that 
$v_j$ is a gathering-candidate node (problem (1) is solved). 
Intuitively, this is because 
$a_i$ does not observe at least $(10g-4) - (8g-4) = 2g$ IDs and this means that 
at least $2g-1$ agents existing in $a_i$'s backward direction also do not observe the necessary number of IDs and 
they eventually stay at $a_i$'s node. 

On the other hand, if $a_i$ observed at least  $8g-3$ IDs,
it uses the first $8g-3$ IDs for comparison and 
considers the $(4g-1)$-st ID as its own ID.   
Then, this situation is similar to one that 
$a_i$ compares its ID with $4g-2$ forward IDs and $4g-2$ backward IDs.
Hence, if the $(4g-1)$-st ID is the minimum among the $8g-3$ IDs,
$a_i$ sets  $v_j.\textit{\textit{candi}} = \textit{true}$ at the current node $v_j$. 
Then, since $k\ge 8g-3$ holds, 
all the $8g-3$ IDs are distinct and thus 
$4g-2$ agents existing in $a_i$'s backward direction can 
recognize $a_i$'s staying node as the nearest gathering-candidate node $v_\textit{candi}$ in the forward direction
when they observed at least $8g-3$ IDs. 
Thus, the $4g-1$ agents in total ($a_i$ and the $4g-2$ agents) try to move forward and stay at $v_\textit{candi}$
(the detail is explained in the next subsection).
Then, when some link continues to be missing,
the $4g-1$ agents are partitioned into two groups and 
at least one group has $2g$ agents (problem (2) is solved).

The pseudocode of the semi-selection phase is described in Algorithm \ref{algo:semiSelection}. 
Concerning the semi-selection phase, we have the following lemma. 

\begin{algorithm}[t!]
	\caption{The behavior of agent $a_i$ in the semi-selection phase 
		($v_j$ is the current node of $a_i$.)}
	\label{algo:semiSelection}

	\Main{}
	$v_j.\textit{id}:=a_i.\textit{id}$, $a_i.\textit{ids}[a_i.\textit{nIDs}] := v_j.\textit{id}$\;
	$a_i.\textit{nIDs} := a_i.\textit{nIDs}+1$, $v_j.\textit{nAgents} := v_j.\textit{nAgents} + 1$\;
	
	
	\While {$a_i.\textit{rounds} \le 3n$}
	{		
		\If{$(a_i.\textit{nIDs} < 10g-4) \land (v_j.\textit{nAgents} < 2g)$} 
		{
			$v_j.\textit{nAgents} := v_j.\textit{nAgents} - 1$\;
			Try to move from the  current node  $v_j$ to the forward node $v_{j+1}$\;
			\If{$(a_i$ reached $v_{j+1}$ (that becomes new $v_j$)) $\land$ $(v_j.\textit{id} \neq \perp) $}
			{
				$a_i.\textit{ids}[a_i.\textit{nIDs}] := v_j.\textit{id}$, $a_i.\textit{nIDs} := a_i.\textit{nIDs}+1$\;
			}
			$v_j.\textit{nAgents} := v_j.\textit{nAgents} + 1$,  $a_i.\textit{rounds} := a_i.\textit{rounds}+1$\;
		}
	}
	
	\If{$(v_j.\textit{nAgents}\ge 2g) \lor	((a_i.\textit{nIDs} \ge 8g-3) \land 		 
		(\forall h\in [0,8g-2]\setminus \{4g-2\}; a_i.\textit{ids}[4g-2]<a_i.\textit{id}[h]))$}
	{
		$v_j.\textit{candi} := \textit{true}$\;
		Terminate the semi-selection phase and enter the semi-gathering phase \;
	}

\end{algorithm}

\begin{lemma}
	\label{lem:semi-selection}
	After finishing  Algorithm \ref{algo:semiSelection}, 
	there exists at least one node $v_j$ with $v_j.\textit{candi} = \textit{true}$.
	
\end{lemma}

\begin{proof}
	Let $a_\textit{min}$ be the agent with minimum ID among all agents and 
	$a_i$ be the $(4g-2)$-nd backward agent of $a_\textit{min}$.
	We consider the cases that 
	the value of $a_i.\textit{nIDs}$ after executing Algorithm \ref{algo:semiSelection} is 
	(a) less than $8g-3$ and (b) at least $8g-3$  in this order. 
	First, (a) if $a_i.\textit{nIDs} < 8g-3$ holds, 
	let $g' = (10g-4) -a_i.\textit{nIDs}$ be the number of IDs that $a_i$ could not observe and 
	$a_{i-1}, a_{i-2}, \ldots , a_{i-(g'-1)}$ be the 1-st, 2-nd, $\ldots$ , $(g'-1)$-st backward agents of $a_i$.
	Then, since $(g'-1) + a_i.\textit{nIDs} = ((10g-4)-a_i.\textit{nIDs}) - 1 + a_i.\textit{nIDs} = 10g-5 < 10g-4$,
	agent $a_{i-(g'-1)}$ does not observe the required number $10g-4$ of IDs. 
	Thus, $a_{i-1}, a_{i-2}, \ldots , a_{i-(g'-1)}$ also observed less than $10g -4$ IDs and 
	they stay at the same node ($a_i'$s node) by the similar discussion of Lemma \ref{lem:more2gSelection}.
	Since $g'-1\ge (10g-4) -(8g-4)-1 = 2g-1$ holds, at least $2g$ agents (including $a_i$) stay at the same node $v_j$ and thus $v_j.\textit{candi}$ is set to true.	
	Next, (b) if $a_i.\textit{nIDs} \ge 8g-3$ holds, 
	$a_i$ recognizes that $a_\textit{min}$'s ID is its own ID and the ID is the minimum among the $8g-3$ IDs.
	Hence, $a_i$ sets $v_j.\textit{candi} = true$ at the current node $v_j$.
	Therefore, the lemma follows. 	
\end{proof}

\subsection{Semi-gathering phase}
In this phase, agents aim to make a configuration such that 
at least $2g$ agents exist at some node.
By Lemma \ref{lem:semi-selection},
there exists at least one gathering-candidate node $v_j$ with $v_j.\textit{candi} = \textit{true}$
at the end of the semi-selection phase.
In the following, we call such a candidate node $v_\textit{candi}$.
Then, if less than $2g$ agents exist at $v_\textit{candi}$, 
$4g-2$ agents in total that already stay at $v_\textit{candi}$ and exist in $v_\textit{candi}$'s backward direction try to stay at $v_\textit{candi}$.
Concretely, in this phase, for $3n$ rounds 
each agent tries to move forward until it stays $v_\textit{candi}$ or 
at least $2g$ agents exist at the current node. 
Then, due to link-missings, it is possible that 
only less than $2g$ agents gather at $v_\textit{candi}$ 
after the movement.
In this case, we can show by the similar discussion of Lemma \ref{lem:more2gSelection} that 
all the agents that do not reach $v_\textit{candi}$ 
among the $4g-2$ agents stay at the same node. 
Then, the $4g-1$ agents (the $4g-2$ agents and the agent originally staying at $v_\textit{candi}$) are partitioned into two groups and 
at least one group has at least $2g$ agents in any partition.
Thus, agents can make a configuration such that at least $2g$ agents exist at some node. 

The pseudocode of the semi-gathering phase is described in Algorithm \ref{algo:semiGathering}.
Note that, during the movement, when agents are blocked few times and they do not reach $v_\textit{candi}$, 
agents may require the total number of  more than $O(gn)$ moves.
To avoid this, each agent stops moving when 
it observed $4g-1$ IDs even if it does not reach $v_\textit{candi}$ (line 2).

\begin{algorithm}[t!]
	\caption{The behavior of agent $a_i$ in the semi-gathering phase 
		($v_j$ is the current node of $a_i$.)}
	\label{algo:semiGathering}       
	
	\Main{}
	$a_i.\textit{rounds} := 1, a_i.\textit{nIDs} := 1$\;
	\While{($a_i.\textit{rounds} \le 3n$) $\land$ ($a_i.\textit{nIDs}\neq 4g-1)$}
	{
		\If{$v_j.\textit{candi} = \textit{false}$}
		{
			$v_j.\textit{nAgents} := v_j.\textit{nAgents} -1$\;
			Try to move from the current node $v_j$ to the forward node $v_{j+1}$\;
			\textbf{if} ($a_i$ reached $v_{j+1}$ (that becomes new $v_{j}$)) $\land$ ($v_j.\textit{id}\neq\perp$) \textbf{then} $a_i.\textit{nIDs} := a_i.\textit{nIDs} +1$\;
			
			$v_j.\textit{nAgents} := v_j.\textit{nAgents} + 1$\;
			\textbf{if} $v_j.\textit{nAgents}\ge 2g$ \textbf{then} $v_j.\textit{candi}:=\textit{true}$\;
			
		}
		$a_i.\textit{rounds} = a_i.\textit{rounds} +1$\;
	}
	Terminate the semi-gathering phase and enter the achievement phase\;
	
\end{algorithm}

Concerning the semi-gathering phase, we have the following lemma.

\begin{lemma}
	\label{lem:semi-gathering}
	After finishing the semi-gathering phase,
	there exists at one node $v_j$ with $v_j.\textit{nAgents}\ge 2g$.
\end{lemma}

\begin{proof}
	We consider a configuration such that 
	there exists no node with at least $2g$ agents at the beginning of the semi-gathering phase. 
	By Lemma \ref{lem:semi-selection}, 
	there exists at least one node $v_j$ with $v_j.\textit{candi} = \textit{true}$,  and 
	$4g-2$ agents in total that already stay at $v_j$ and exist in $v_j$'s backward direction try to stay at $v_j$ by Algorithms \ref{algo:semiSelection} and \ref{algo:semiGathering}.
	Then, by the similar discussion of the proof of Lemma \ref{lem:more2gSelection},
	after executing Algorithm \ref{algo:semiGathering} for $3n$ rounds,
	all agents among the $4g-2$ agents that do not reach $v_j$ stay at the same node.
	Thus, the $4g-1$ agents (the $4g-2$ agents and the agent originally staying at $v_j$) are partitioned into two groups 
	and at least one group has at least $2g$ agents in any partition.
	Therefore, the lemma follows. 
\end{proof}

\subsection{Achievement phase}
In this phase, agents aim to achieve $g$-partial gathering.
By Lemma \ref{lem:semi-gathering}, there exists at least one node with at least $2g$ agents as in Section \ref{sec:more3g}.
The difference from Section \ref{sec:more3g} is that 
there may exist more than two nodes with agents and 
there may exist several nodes each of which has at least $2g$ agents. 
Also from this situation, agents can  achieve $g$-partial gathering  using the same method as that in Section \ref{sec:more3g}, that is, 
(1) agents staying at a node with at least $2g$ agents  are partitioned into a forward group and a backward group and
they try to move forward and backward respectively, and  
(2) when a forward group (resp., a backward group) visits a node with less than $g$ agents,
the less than $g$ agents join the forward group (resp.,  the backward group).

An example is given in Fig.\,\ref{fig:achievement}.
In Fig.\,\ref{fig:achievement} (a),
there exist two nodes $v_p$ and $v_q$ each of which has  $6\,(=2g)$ agents. 
Hence, a forward group $A_{f_p}$ and a backward group $A_{b_p}$ 
(resp., $A_{f_q}$ and $A_{b_q}$) start moving from node $v_p$
(resp., from node $v_q$).
From (a) to (b), 
$A_{b_p}$ reaches node $v_\ell$ with less than $g$ agents, and the less than $g$ agents join $A_{b_p}$ and 
try to move backward. 
From (b) to (e), we assume that $A_{f_q}$ continues to be blocked by a missing link.
From (b) to (c),
$A_{f_p}$ and $A_{b_q}$ crossed and they recognize the fact by the existence of  flags,
and they terminate the algorithm execution.
From (c) to (d),
$A_{b_p}$ reaches node $v_m$ with less than $g$ agents, and the less than $g$ agents join $A_{b_p}$ and 
try to move backward. 
Then, since the number \textit{num} of agents in the updated $A_{b_p}$ is $7\,(> 2g)$, by using their IDs,
only $g$ agents continue to try moving backward and the remaining $\textit{num}-g$ agents terminate the algorithm execution there
((d) to (e)). 
By  this behavior, during this phase each link is passed by at most $2g$ agents and 
the achievement phase can be achieved with the total number of $O(gn)$ moves. 
From (e) to (f), $A_{f_q}$ and $A_{b_p}$ reach some node simultaneously and recognize the fact by flags, 
and they terminate the algorithm execution
and agents achieve $g$-partial gathering. 

\begin{figure}[t!]
	\centering
	\includegraphics[scale=0.575]{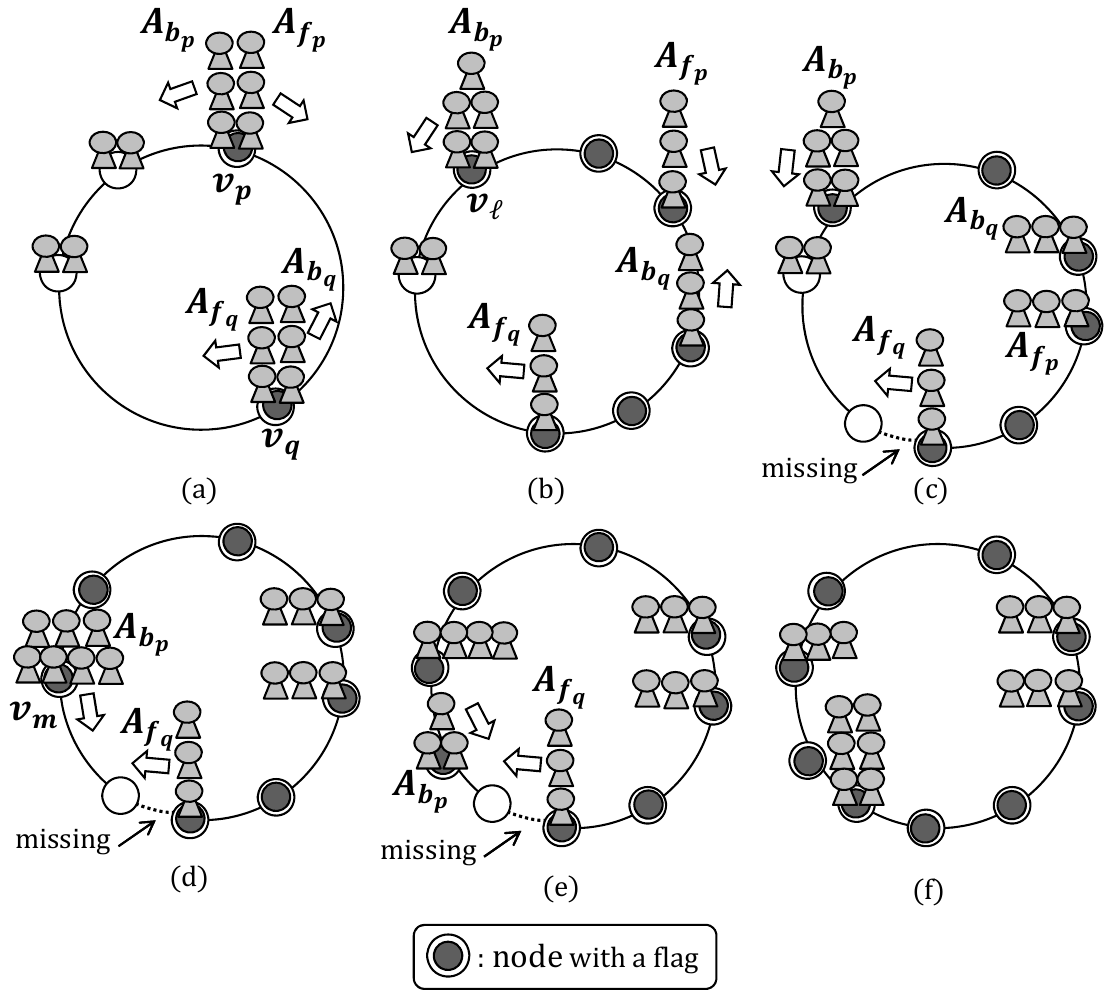}
	\caption{An execution example of the achievement phase when $k\ge 8g-3$ ($g=3$).}
	\label{fig:achievement}
\end{figure} 

Concerning the achievement phase, we have the following lemma.

\begin{lemma}
	\label{lem:achievement}
	After executing the achievement phase, agents achieve $g$-partial gathering. 
\end{lemma}

\begin{proof}
	By Lemma \ref{lem:semi-gathering}, at the beginning of the achievement phase, 
	there exists at least one node with at least $2g$ agents. 
	Thus, the at least $2g$ agents can execute \textit{More2}() 
	(Procedure \ref{algo:more}) and start moving forward or backward. 
	For the rest, 
	we can show the correctness by the exact same discussion of the proof of Lemma \ref{lem:more3gGathering}.
	Therefore, the lemma follows. 
\end{proof}

We have the following theorem for the proposed algorithm.

\begin{theorem}
	\label{theo:largeK}
	When $k\ge 8g -3$ holds, 
	the proposed algorithm solves the $g$-partial gathering problem in dynamic rings with 
	$O(n)$ rounds and the total number of $O(gn)$ moves. 
\end{theorem}

\begin{proof}
	By Lemmas \ref{lem:semi-selection}, \ref{lem:semi-gathering}, and \ref{lem:achievement},
	the proposed algorithm can solve the $g$-partial gathering problem. 
	In the following, we analyze the time complexity and the required total number of agent moves. 
	
	First, in the semi-selection phase, for $3n$ rounds
	each agent $a_i$ tries to move forward until it observes $10g -3$ IDs or at least $2g$ agents exist at the current node,
	which requires $O(n)$ rounds. 
	In addition, since $a_i$ stops moving when it observes $10g -3$ IDs, 
	in this phase each link is passed by at most $10g-3$ times.
	Hence, the move complexity of the semi-selection phase is $O(gn)$.
	
	Next, in the semi-gathering phase, for $3n$ rounds 
	each agent at a node with less than $2g$ agents tries to move forward until 
	it stays at a gathering-candidate node $v_\textit{candi}$, a node with at least $2g$ agents, or it observes $4g-1$ IDs, which requires $O(n)$ rounds.
	In addition, since each link is passed by at most $4g-1$ times, 
	the move complexity of the semi-gathering phase is $O(gn)$.
	
	Finally, in the achievement phase, 
	agents staying at a node with at least $2g$ agents move for achieving $g$-partial gathering.
	As discussed in Lemmas \ref{lem:more3gGathering} and \ref{lem:achievement},
	a forward group of a backward group can visit a next node at each round, and hence
	the achievement phase finishes within $n$ rounds. 
	In addition, each forward or backward group comprises at most $2g$ agents, and hence 
	each link is passed by at most $4g$ times.
	Thus, the move complexity of the achievement phase is $O(gn)$.
	Therefore, agents can solve the $g$-partial gathering problem with $O(n)$ rounds and the total number of $O(gn)$ moves,
	and the theorem follows. 
\end{proof}

\section{Conclusion} 
In this paper, we considered the $g$-partial gathering problem in bidirectional dynamic rings and 
considered the solvability of the problem and the time and move complexity, 
focusing on the relationship between values of $k$ and $g$. 
First,  we showed that agents cannot solve the problem when $k \le 2g$.
Second, we showed that  the problem can be solved   with 
$O(n\log g)$ time and the total number of $O(gn\log g)$ moves when $2g+1\le k \le 3g-2$.
Third, 
we showed that the problem can be solved with $O(n)$ time and the total number of $O(kn)\,(=O(gn))$ moves
when $ 3g-1\le k\le 8g-4$.
Finally,
we showed that the problem can be solved with $O(n)$ time and the total number of $O(gn)$ moves  when $k\ge 8g - 3$.
These results show that the partial gathering problem can be solved also in dynamic rings when $k\ge 2g+1$.
In particular, when $k\ge 3g-1$, the time complexity $O(n)$ of the proposed algorithms is asymptotically optimal and 
the total number $O(gn)$ moves is also asymptotically optimal. 

Future works are as follows.
First, we consider whether or not agents can solve the problem with the total number of moves smaller than $O(gn\log g)$ in the case of $2g+1\le k\le 3g-2$.
Also, we will consider the solvability of the problem for agents with weaker capabilities, e.g., agents without distinct IDs, without chirality, or 
agents that behave semi-synchronously or asynchronously. 
In any of the above cases, we conjecture that 
agents cannot solve the problem or require more total number of moves than the proposed algorithms. 

\bibliographystyle{unsrt}
\bibliography{refP}   

\end{document}